\documentclass[journal]{IEEEtran}
\usepackage{graphicx}
\usepackage{algorithm}
\usepackage{algorithmic}
\usepackage{multicol}
\usepackage{calc}
\usepackage{subfigure}
\usepackage{hyperref}

\usepackage{url}
\usepackage{cite}
\usepackage[utf8]{inputenc}
\usepackage[english]{babel}
\usepackage{multicol}
\usepackage{amsmath}
\usepackage{amsfonts}
\usepackage{amssymb}
\usepackage{pdflscape}
\usepackage{multirow}
\usepackage{mathrsfs}
\usepackage{amsthm}
\usepackage{ragged2e}
\usepackage{verbatim}
\usepackage{amsthm}
\usepackage{color}
\usepackage{wrapfig}
\usepackage{bm}

\def \Z {{\mathbb{Z}}}

\def \x {{\bf x}}

\def \n {{\mathbb{N}}}

\newtheorem{theorem}{Theorem}
\newtheorem{lemma}{Lemma}
\newtheorem{proposition}{Proposition}
\newtheorem{conjecture}{Conjecture}

\theoremstyle{definition}
\newtheorem{example}{Example}

\newtheorem{assumption}{Assumption}
\theoremstyle{definition}
\newtheorem{definition}{Definition}

{\tiny  
  \setlength{\abovedisplayskip}{6pt}
  \setlength{\belowdisplayskip}{\abovedisplayskip}
  \setlength{\abovedisplayshortskip}{0pt}
  \setlength{\belowdisplayshortskip}{3pt}
}

\begin{document}

\title{On Naisargik Images of Varshamov-Tenengolts and Helberg Codes}

\author{Kalp Pandya, 
    Devdeep Shetranjiwala, Naisargi Savaliya
    and Manish K. Gupta,~\IEEEmembership{Senior Member,~IEEE}\thanks{Corresponding Author's Email:mankg@guptalab.org}\\
    Dhirubhai Ambani Institute of Information and Communication Technology\\
Gandhinagar, Gujarat, 382007 India\\
Email: kalppandya94@gmail.com, devdeep0702@gmail.com, savaliyanaisargi@gmail.com, mankg@guptalab.org}
    
\maketitle

\begin{abstract}
The VT and Helberg codes, both in binary and non-binary forms, stand as elegant solutions for rectifying insertion and deletion errors. Within this manuscript, we consider the quaternary versions of these codes. It is well known that many optimal binary non-linear codes like Kerdock and Prepreta can be depicted as Gray images (isometry) of codes defined over $\Z_4$. Thus a natural question arises: Can we find similar maps between quaternary and binary spaces which gives interesting properties when applied to the VT and Helberg codes. We found several such  maps called Naisargik (natural) maps and we study the images of quaternary  VT and Helberg  codes under these maps. Naisargik and inverse Naisargik images gives interesting error-correcting properties for VT and Helberg codes. If two Naisargik images of VT code generates an intersecting one deletion sphere, then the images holds the same weights. A quaternary Helberg code designed to correct $s$ deletions can effectively rectify $s+1$ deletion errors when considering its Naisargik image, and $s$-deletion correcting binary Helberg code can corrects $\lfloor\frac{s}{2}\rfloor$ errors with inverse Naisargik image. 

\end{abstract}

\begin{IEEEkeywords}
 DNA Storage, DNA codes , Helberg codes , VT Codes, Gray map, Naisargik map, Quternary, Deletion errors
\end{IEEEkeywords}
\vspace{-.2in}
\section{Introduction}
\label{intro}

The significance of DNA storage lies in its capacity to address the limitations of conventional storage methods through high data density, stability, and environmental sustainability. This transformative technology efficiently manages the escalating demands of digital data storage, offering scalability and improved data security. DNA storage contains the data in form of sequence, consisting nucleotides adenine (A), cytosine (C), guanine (G), and thymine (T). Typical DNA storage follows the encoding of digital data in ACGT sequence and the decoding it back to it's original form. There are various encoding - decoding methods available for DNA storage data called codecs. The optimal codec algorithm is the one which corrects more number of errors. Codec does not only ensure the efficacy of data storage and retrieval but also systematically tackle potential errors arising from insertion or deletion during transmission. This exploration contributes to the establishment of a robust and dependable system at the convergence of information technology and biological sciences.

An excellent class of VT codes exists for special cases of single insertion or deletion errors. R. L. Varshamov and G. M. Tenengolts \cite{vtcodeoriginal} proposed binary VT code for asymmetric errors. Levenshtein \cite{Levenshtein1965BinaryCC} reconstructed the VT code and provided a linear-time decoding algorithm for the VT codes. Both VT constructions have the $\log(n) + O(1)$ redundant bits, where n represents length of the codewords. Later, Tenengolts \cite{nonbinaryVT} introduced non-binary VT code, which is also known as q-ary VT code, where q \(>\) 2. He also proposed a linear-time single error-correcting method along with an efficient encoder. Abroshan et al. \cite{encodingalgo} introduced a systematic encoder designed for non-binary VT codes, offering linear complexity. Additionally, they established a novel lower bound on the size of q-ary VT codes. Zihui et al. \cite{soft} introduces an innovative soft-in soft-out decoding algorithm tailored for VT codes applied to multiple received DNA strands. This algorithm addresses insertion, deletion, and substitution errors commonly encountered in DNA-based data storage. Sloane \cite{sloane2002singledeletioncorrecting} summarized the binary VT code and proved some import theorem about the length of the VT code $|VT_a(n)|$, where a is residue and n represents the length of the codeword. 

Helberg \cite{helberg} extended the VT code to correct the multiple insertion or deletion error. The number of insertions/deletions is denoted by s; s = 1, representing that code can correct a single error. Levenshtein \cite{Levenshtein1965BinaryCC} code are based on s = 1. The upper bound is $s \geq d_{min} - 1$, where $d_{min} $ is the minimum Hamming distance between codes. The lower bound is $s \leq \Delta w_{min} - 1$, where $\Delta w_{min}$ is the minimum weight of codewords. Abdel-Ghaffar \cite{le2023new} compares the Helberg code to the Levenshtein code and demonstrates that the Helberg code outperforms the Levenshtein code in terms of error correction capabilities and also provides proof for the same. Tuan \cite{abdelghaffar2012} et al. describe a new construction of multiple insertion and deletion correcting codes for non-binary bits by generalizing Helberg's method for constructing binary codes utilizes number-theoretic principles. Additionally, it outlines a linear decoding technique aimed at rectifying multiple deletions. They also provide theoretical bounds on the code's performance and contrast it with other codes capable of rectifying insertion and deletion errors. The paper concludes with potential applications of these new codes.

DNA storage data consists 4 types nucleotides. So, for each  nucleotide we need mapping of digital bits which can be done by elements from $\Z_4$. In 1994, many non-linear binary optimal codes were constructed using an isometry (called Gray map) from $\Z_4 \rightarrow \Z_2^2$ \cite{Z4nonlinear}. In partiicular, for example, non-linear Kerdock and Preparata codes become linear over $\Z_4$, when transformed from $\Z_2$ \cite{Z4nonlinear}. One natural question that we can ask is the following: Can we construct natural (we call such class of maps as naisargik maps) maps from quaternary space to binary space that can preserve the deletion errror correction property. Out of total 24 possible maps from $\Z_4 \rightarrow \Z_2^2$, we found eight maps for VT codes and only one map for Helberg code having interseting properties. In particular,...  In this paper, our emphasis is directed towards discovering such naisargik maps over VT codes and Helberg codes as these are the powerful single and multiple error correcting codes respectively. There are total  4! maps possible for mapping $\Z_4 \rightarrow \Z_2^2$. Out of which few are naisargik.

The study on $4! = 24$ maps shows that some of these maps gives interesting properties over VT codes and Helberg codes. Properties of VT code after applying map are if delete 2 equal consecutive bits then VT code can be two error correcting code, for one deletion, if the sequences generate the intersecting deletion sphere then these sequences have same number of ones and zeros. VT code holds this properties for eight different maps. As discussed above Helberg code is more powerful than VT code and map enhances the performance of Helberg code as the results shows that s deletion correcting codes over $\Z_{4}$ gives s+1 deletion correcting capability over $\Z_2^2$. Similarly, reverse mapping from $\Z_2^2$ to $\Z_{4}$ gives the s and s-1 deletion error correcting capability. The mathematical proofs and data generated for the observation are discussed thoroughly in below sections. More information about multiple error correction, different types of errors and mapping is discussed in \cite{8022906,erlich2017dna,hu2018optimal,helbergmoredetails,article1more,article2more,article3more}.

The structure of the paper is outlined as follows: Section II covers the theoretical background of the deletion sphere algorithm. Section III introduces the Naisargik map, inverse Naisargik map, and binary codes like reduction code and torsion code. Section IV shows the study of the Naisargik map and the inverse of the Naisargik map on VT code. Section V shows the study of the Naisargik map and the inverse of the Naisargik map on the Helberg code. Section VI concludes the study of Naisargik maps. Section VII contains a reference to the data used for the study. 

\section{Preliminaries and Background}
For a positive non-zero integer $q$, denote the ring of integers modulo q as $\Z_q$. A code $\mathscr{C}$ is a subset of $\Z_q^n$ containing $q^n$ codewords. In addition, if  $\mathscr{C}$ is also an additive subgroup of $\Z_q^n$, we call it a linear code. For insertion/deletion channels, quaternary (q=4) codes are helpful for DNA storage applications \cite{Levenshtein1965BinaryCC}. The deletion distance for $x, y \in \Z_q^n$ can be defined as $d_{e}(x, y)$, i.e. the minimum amount of deletion required for changing $x$ to $y$. $\forall x, y \in \mathscr{C}$, if $d_{e}(x, y) = 0$, this means $x = y$. For any given codeword $x \in \mathscr{C}$, if we remove $s$ symbols, we get a subsequence of $x$. The collection of all such subsequences of $x$ is known as the deletion sphere for the codeword $x$. The deletion sphere for a specific $s$ is defined as:
\begin{equation}
    D_s(x) = \{w | \text{w is an (n-s) length subsequence of x} \}.
\end{equation}
A code $\mathscr{C}$ can correct $s$ errors only if the spheres $D_s(x)$ and $D_s(y)$ do not intersect, i.e., $D_s(x) \cap D_s(y) = \emptyset$. If $\forall \; x, y \in \mathscr{C}$, $x \neq y$ and $D_s(x) \cap D_s(y) = \emptyset$, then the it's called s deletion error correcting code.

For example, consider $\mathscr{C} = \{000101, 010000\}$. If the deletion spheres for both codewords are distinct, we can say that it can correct $s$ errors. If we define the deletion sphere for both sequences for two deletions, then the deletion spheres for 000101 and 010000 are defined as follows:

\begin{align*}
   D_2(000101) &= \{(0101), (0001), (0011), (0010)\},\\
    D_2(010000) &= \{(0000), (1000), (0100)\}.  
\end{align*}

Since $D_2(000101) \cap D_2(010000) = \emptyset$, the codewords are $s = 2$ deletion error correcting.


Refer to Algorithm \ref{sphere} for generating deletion sphere for $s$ deletions. Let us explore the mapping in detail to understand how VT and Helberg codes exhibit interesting properties by combining the deletion sphere as a decoding algorithm and mapping $\Z_4 \rightarrow \Z_2^2$.
\section{Naisargik Maps for $\Z_4\rightarrow \Z_2^2$ and Codes}
Let $\Z_4 = \{0,1,2,3\}$ denote the set of nucleotides of the DNA, each represented by a single element. Correspondingly, let $\Z_2^2 = \{00,01,10,11\}$ denote the set of 2-bit digital data, where each element comprises binary digits 0 and 1. 4! = 24 possible mappings exist between elements of $\Z_4$ and $\Z_2^2$. Among these mappings, nine are identified as Naisargik maps, exhibiting noteworthy properties when applied to the VT and Helberg codes. Table \ref{tab:mapping} presents these Naisargik maps formally.
\renewcommand{\arraystretch}{1.3}
\begin{table}[h]
\centering
\caption{Naisargik maps for $\Z_4\rightarrow \Z_2^2$}
\label{tab:mapping}
\begin{tabular}{|c|c|}
\hline
 $\phi_i$        &   Maps\\ 
 \hline
$\phi_1$ & $\{0\rightarrow 00,\; 1\rightarrow 10,\;2\rightarrow 11,\;3\rightarrow  01\}$ \\ \hline
$\phi_2$ & $\{0\rightarrow 01,\; 1\rightarrow 00,\;2\rightarrow 10,\;3\rightarrow  11\}$ \\ \hline
$\phi_3$ & $\{0\rightarrow 01,\; 1\rightarrow 11,\;2\rightarrow 10,\;3\rightarrow  00\}$ \\ \hline
$\phi_4$ & $\{0\rightarrow 11,\; 1\rightarrow 01,\;2\rightarrow 00,\;3\rightarrow  10\}$ \\ \hline
$\phi_5$ & $\{0\rightarrow 11,\; 1\rightarrow 10,\;2\rightarrow 00,\;3\rightarrow  01\}$ \\ \hline
$\phi_6$ & $\{0\rightarrow 10,\; 1\rightarrow 00,\;2\rightarrow 01,\;3\rightarrow  11\}$ \\ \hline
$\phi_7$ & $\{0\rightarrow 10,\; 1\rightarrow 11,\;2\rightarrow 01,\;3\rightarrow  00\}$ \\ \hline
$\phi_8$ & $\{0\rightarrow 00,\; 1\rightarrow 01,\;2\rightarrow 11,\;3\rightarrow  10\}$ \\ \hline
$\phi_9$ & $\{0\rightarrow 11,\; 1\rightarrow 01,\;2\rightarrow 10,\;3\rightarrow  00\}$ \\ \hline
\end{tabular}
\end{table}

VT code uses eight Naisargik maps from $\phi_1 $ to $\phi_8$. At the same time, the properties of the Helberg code are observed over Naisargik map $\phi_9$. We will use the symbol $\phi$ for Naisargik maps for simplicity. In section IV, $\phi$ represents all the Naisargik maps from $\phi_1 $ to $\phi_8$ and in section V, $\phi$ represents the Naisargik map $\phi_9$. Similarly, $\phi^{-1}$ represents the inverse of all the Naisargik maps. VT code refers $\phi^{-1}$ as the inverse of Naisargik maps from $\phi_1 $ to $\phi_8$ and Helberg code referse $\phi^{-1}$ as inverse of Naisargik map $\phi_9$.

\section{Varshamov-Tenengolts Code (VT Code)}
VT code is a type of error-correcting code in coding theory and information theory, proposed by R. L. Varshamov and G. M. Tenengolts\cite{vtcodeoriginal}. These are a type of linear block, designed to have minimal Hamming distance between code words, allowing them to correct errors effectively.

There exist two distinct categories of VT codes: Binary VT codes and q-ary VT codes. In the context of a binary VT code, a codeword x corresponds to a residue denoted as 'a'. Conversely, in the q-ary VT code, a codeword x corresponds to a residue represented as (a, b). The process of VT coding involves generating residues and their associated sets of codewords. Below are the definitions of these VT codes.
\begin{definition}[Binary VT Code:]\label{Binary_VT_CODE}
    For VT residue $a \in \Z_{n+1}$ the equation of binary version of VT code\cite{nguyen2022bit} is,
\begin{equation}    
VT_a (n) = \{ x \in \Z_2^n : a\; (mod\; n+1) = \sum_{i=1}^n ix_i \}.
\end{equation}
\end{definition}
Every set of VT code corrects unit deletion/insertion error where 0 $\leq$ a $\leq$ n. Each set of residue will have at most $\frac{2^n}{n+1}$ sequences. If n+1 is a power of 2, then the length of the VT residue set is exactly $\frac{2^n}{n+1}$ \cite{vtcodenumbers}.

\begin{definition}[q-ary VT Code]\label{qary_Definition}
    Consider $x \in \mathbb{Z}_q^n$. The monotonicity sequence or signature sequence $\alpha(x)$ of length $n-1$ is defined as follows \cite{nguyen2022bit}:
    \[
    \alpha(x)_i = 
    \begin{cases}
      1, & \text{if } x_i \leq x_{i+1} \\
      0, & \text{otherwise}
    \end{cases}
    \]

For q, n $\geq$ 1, $a \in \Z_n$ and $b \in \Z_q$, q-ary VT code equation is

\begin{equation} 
\label{eq4}
\begin{split}
VT_{a,b} (n; q) &\triangleq \{ x \in \Z_q^n : \alpha(x) \in VT_a(n-1) \\
&\quad \text{and} \; \sum_{i=1}^n x_i \; (\text{mod} \; q) = b \}.
\end{split}    
\end{equation}
\end{definition}

This code will generate output for every possible a and b combination which is of size at least $\frac{q^n}{(qn)}$\cite{nguyen2022bit}.
Let us study the influence of Naisargik maps in the context of VT code defined over $\Z_4$. VT code defined over $\Z_4$ is called the Quaternary VT code.\\

\textbf{Note:} In this section, Naisargik maps $\phi_1$ to $\phi_8$ (Table \ref{tab:mapping}) are referred to as $\phi$.

\subsection{\textbf{Naisargik Image of Quaternary VT Code}}
Consider codeword $x \in \Z_4^n$. Map x with the help of $\phi$ such that $\phi(x) \in \Z_2^{2n}$. Quaternary VT code $VT_{a,b} (n;4)$ is the subset of $\Z_4^n$ and Naisargik image of Quaternary VT code $\phi(VT_{a,b} (n;4))$ is the subset of $\Z_2^{2n}$. Mathametically,
$$VT_{a,b} (n;4) \in \Z_4^n \implies \phi(VT_{a,b} (n;4)) \in \Z_2^{2n}$$

The observation states that Naisargik images of Quaternary VT code give the intersecting sphere for one deletion error. If the codewords x and y belong to the same residue set and give the intersecting deletion sphere, then the x and y have same weight. Let us look at the example of the approach.

\begin{example}
    Consider Quaternary VT code of length n = 4 and $\phi = \phi_8$. From definition \ref{qary_Definition}, generate all the possible residue sets (a,b) and corresponding set of codewords for $VT_{a,b} (4;4)$. The residue values and the number of codewords corresponding to the residue are given in Table \ref{qaryexample}. Consider residue value (1,2) and apply Naisargik map $\phi$ on each codeword corresponding to residue (1,2). Table \ref{tab:qtobexample} shows the codeword and corresponding mapped codeword. 

\begin{table}[h]
\centering
\caption{Naisargik images of $VT_{1,2} (4;4)$ with residue (1,2)}
\label{tab:qtobexample}
\begin{tabular}{|c|c|}
\hline
\textbf{Codewords of $VT_{1,2} (4;4)$} & \textbf{Codewords of $\phi(VT_{1,2} (4;4))$} \\
\hline
(0, 3, 2, 1) & (0, 0, 1, 0, 1, 1, 0, 1) \\
\hline
(1, 0, 0, 1) & (0, 1, 0, 0, 0, 0, 0, 1) \\
\hline
(1, 0, 2, 3) & (0, 1, 0, 0, 1, 1, 1, 0) \\
\hline
(1, 3, 2, 0) & (0, 1, 1, 1, 1, 0, 0, 0) \\
\hline
(2, 0, 0, 0) & (1, 1, 0, 0, 0, 0, 0, 0) \\
\hline
(2, 0, 1, 3) & (1, 1, 0, 0, 0, 1, 1, 0)\\
\hline
(2, 0, 2, 2) & (1, 1, 0, 0, 1, 1, 1, 1)\\
\hline
(2, 1, 1, 2) & (1, 1, 0, 1, 0, 1, 1, 1)\\
\hline
(2, 3, 1, 0) & (1, 1, 1, 0, 0, 1, 0, 0)\\
\hline
(3, 0, 0, 3) & (1, 0, 0, 0, 0, 0, 1, 0)\\
\hline
(3, 0, 1, 2) & (1, 0, 0, 0, 0, 1, 1, 1)\\
\hline
(3, 1, 1, 1) & (1, 0, 0, 1, 0, 1, 0, 1)\\
\hline
(3, 1, 3, 3) & (1, 0, 0, 1, 1, 0, 1, 0)\\
\hline
(3, 2, 2, 3) & (1, 0, 1, 1, 1, 1, 1, 0)\\
\hline
\end{tabular}
\end{table}
\par The decoding algorithm generates a deletion sphere for each codeword. VT code is only one deletion error-correcting, which means it generates a non-intersecting deletion sphere for codewords of the same residue. However, Naisargik images of VT codes generate the intersecting deletion sphere. Table \ref{tab:residuea} shows the deletion sphere of $\phi(VT_{1,2} (4;4))$ which corresponds to residue (1,2).

\par The observation states that taking image of Quaternary VT code using map do not hold the property of deletion error correction. However following proposition holds after mapping. The proposition hold true for the eight maps mentioned in Table \ref{tab:mapping}.

\end{example}

From the table \ref{tab:qtobexample}, we can formulate relations between $x$ and $X = \phi(x)$ and between Quaternary VT code parameters (recall definition \ref{qary_Definition}) $a$, $b$ and $X$. In the following lemmas, we state these relations.

\begin{lemma} \label{lemma: conversion}
    For $\phi = \phi_8$ from Table \ref{tab:mapping}, for any $x \in \Z_4^n$ and $X \in \Z_2^{2n}$ if $X = \phi(x)$, then
    \begin{equation}
        x_i = 3X_{2i-1} + X_{2i} - 2X_{2i-1}X_{2i}
    \end{equation}

    where $x_i$ represents $i^{th}$ bit of sequence x.
    \begin{proof}
        Let's assume the following linear relation for a specific bit $i$:
    \[
        x_i = a + bX_{2i-1} + cX_{2i} + dX_{2i-1}X_{2i}
    \]
    where $a$, $b$, $c$, and $d$ are constants to be determined.
    Now, we get four equations for all four inputs:
    \begin{align*}
        \text{For } X_{2i-1} = 0, \; X_{2i} = 0, \; and \; x_i = 0: & \quad a \cdot a = 0 \\
        \text{For } X_{2i-1} = 0, \; X_{2i} = 1, \; and \; x_i = 1: & \quad a \cdot 1 = c \\
        \text{For } X_{2i-1} = 1, \; X_{2i} = 0, \; and \; x_i = 3: & \quad a \cdot 3 = b \\
        \text{For } X_{2i-1} = 1, \; X_{2i} = 1, \; and \; x_i = 2: & \quad a \cdot -2 = d \\
    \end{align*}
    Solving these four equations for four variables, we obtain $a = 0$, $b = 3$, $c = 1$, and $d = -2$. Therefore, the assumed relation holds true for all possible cases, establishing the validity of the lemma. \\
    \end{proof}
\end{lemma}

\begin{example}
Let's verify Lemma \ref{lemma: conversion} from an example:\\
Let $X = 00101101$ and $X = \phi(x)$ where $x \in \Z_4^4$. \\ \\
According to Lemma \ref{lemma: conversion}:
\[
    x_i = 3X_{2i-1} + X_{2i} - 2X_{2i-1}X_{2i}
\]
For $i = 1$: $x_1 = 3 \times 0 + 0 - 2 \times 0 \times 0 = 0$\\
For $i = 2$: $x_2 = 3 \times 1 + 0 - 2 \times 1 \times 0 = 3$\\
For $i = 3$: $x_3 = 3 \times 1 + 1 - 2 \times 1 \times 1 = 2$\\
For $i = 4$: $x_4 = 3 \times 0 + 1 - 2 \times 0 \times 1 = 1$\\ \\
Hence, the 4-bit quaternary sequence $x$ corresponding to $X = 00101101$ is $x = 0321$.
\end{example}
\begin{lemma}
\label{lemma_ax}
    For $\phi = \phi_8$ from Table \ref{tab:mapping}, $x \in \Z_4^n$ and $X \in \Z_2^{2n}$ if $X = \phi(x)$, then

    \[
    \alpha(x)_i = 
    \begin{cases}
      1, & \text{if } x_i \leq x_{i+1} \\
      0, & \text{otherwise}
    \end{cases}
    \]
    
    Here $\alpha(x)_i$ represents the quaternary VT code parameter from definition \ref{qary_Definition}. We can simplify the above equation to the boolean function given below,
    
    \begin{align*}
        \alpha(x)_i &= 1 - X_{2i-1} - X_{2i} + X_{2i}X_{2i+2}\\
        &+ X_{2i}X_{2i-1} + X_{2i-1}X_{2i+1} - X_{2i-1}X_{2i}X_{2i+2}\nonumber \\ 
        &- X_{2i}X_{2i+1}X_{2i+2} - X_{2i-1}X_{2i}X_{2i+1} \nonumber \\
        &- X_{2i-1}X_{2i+1}X_{2i+2} + 2X_{2i-1}X_{2i}X_{2i+1}X_{2i+2}
    \end{align*}
    \begin{proof}
        For any quaternary sequence x if $x_i \leq x_{i+1}$, then $\alpha(x)_i = 1$. All such cases are shown in Table \ref{tab:alphax1}. \\
        

\begin{table}
\centering
\caption{Generalization of quaternary sequences where $\alpha(x)_i = 1$. Here, $x_i$ represents the $i^{th}$ term of the quaternary sequence x, and $X_i$ represents the $i^{th}$ term of the corresponding binary sequence X. $\alpha, \beta \in \Z_2$. The 'Generalized Form' column denotes the pattern generalized from the binary sequences.}
\label{tab:alphax1}
\begin{tabular}{|c|c|c|p{1.2cm}|}
\hline
\textbf{$x_i$,\; $x_{i+1}$} & \textbf{$X_{2i-1}$,\;$X_{2i}$,\;$X_{2i+1}$,\;$X_{2i+2}$} &  & Generalized Form \\ \hline
0 0 & 0 0 0 0 & \multirow{4}{*}{$\implies$} & \multirow{4}{*}{0 0 $\alpha$ $\beta$} \\ \cline{1-2}
0 1 & 0 0 0 1 &                           &                          \\ \cline{1-2}
0 2 & 0 0 1 1 &                           &                          \\ \cline{1-2}
0 3 & 0 0 1 0 &                           &                          \\ \hline
1 1 & 0 1 0 1 &$\implies$                  & 0 1 0 1                  \\ \hline
1 2 & 0 1 1 1 & \multirow{4}{*}{$\implies$} & \multirow{4}{*}{$\alpha$ 1 1 $\beta$} \\ \cline{1-2}
1 3 & 0 1 1 0 &                           &                          \\ \cline{1-2}
2 2 & 1 1 1 1 &                           &                          \\ \cline{1-2}
2 3 & 1 1 1 0 &                           &                          \\ \hline
3 3 & 1 0 1 0 &$\implies$                       & 1 0 1 0                  \\ \hline
\end{tabular}
\end{table}

        Representing generalized form from table \ref{tab:alphax1} in equation,
        \begin{align*}
            \alpha(x)_i &= (1-X_{2i-1})(1-X_{2i}) \\
                        &+ (1-X_{2i-1})X_{2i}(1-X_{2i+1})X_{2i+2} + X_{2i}X_{2i+1}\\ 
                        & + X_{2i-1}(1-X_{2i})X_{2i+1}(1-X_{2i+2})
        \end{align*}

        On solving above equation,

        \begin{align*}
            \alpha(x)_i &= 1 - X_{2i-1} - X_{2i} + X_{2i}X_{2i+2} + X_{2i}X_{2i-1} \nonumber \\ 
            &+ X_{2i-1}X_{2i+1} - X_{2i-1}X_{2i}X_{2i+2} - X_{2i}X_{2i+1}X_{2i+2} \nonumber \\
            &- X_{2i-1}X_{2i}X_{2i+1} - X_{2i-1}X_{2i+1}X_{2i+2} \nonumber \\ &+ 2X_{2i-1}X_{2i}X_{2i+1}X_{2i+2}
        \end{align*}
    \end{proof}
\end{lemma}

It would be tempting to study if $\phi(VT_{(a,b)}(n, 4))$ also shows deletion error correction property. On observing data, we find \(\nexists X \in \phi(VT_{(a,b)}(n, 4))\) that does not show deletion error correction property. However, we observe an interesting conjecture \ref{p2} from our data.

\subsection{\textbf{Error Correction For Naisargik Image of Quaternary VT Code}}

While comprehensive proof for this assertion is currently unavailable, we have provided partial proof substantiated by an exhaustive examination of data.

\begin{assumption} \label{a1}
     Naisargik map $\phi$ denotes all the maps from $\phi_1$ to $\phi_8$ (Table \ref{tab:mapping}).
\end{assumption}

\begin{conjecture}
  \label{p2}
    $\forall X,Y \in \phi(VT_{(a,b)}(n;4))$ and $D_1(X) \cap D_1(Y) \neq \emptyset$, then $w_X = w_Y$, where w represents the weight of binary sequence.
\end{conjecture}  

\textit{\textbf{Justification:}} Recall definition \ref{qary_Definition}, a and b are the parameters of Quaternary VT code. Obviously, if sequences $x,y \in VT_{(a,b)}(n;4)$ are in same residue set then the difference of corresponding VT parameters a and b will be zero, meaning $a_x - a_y = 0$ and $b_x - b_y = 0$, where $a_x$ and $b_x$ are the VT parameters of sequence x and $a_y$ and $b_y$ are VT the parameters of sequence y. For example, if $x = 1320$, then $a_x = 1$ and $b_x = 2$. \\
\par \textbf{(Outline)} For $X,Y \in \Z_4^{2n}$, assume $w_X \neq w_Y$ and $D_1(X) \cap D_1(Y) \neq \emptyset$, where $w_X$ and $w_Y$ are the weights of sequence X and Y, respectively and $D_1(X)\; and \;D_1(Y)$ are the deletion spheres after one bit deletion for sequence X and Y respectively. Let $X \in \phi(VT_{(a,b)}(n;4))$. We aim to establish the assertion $Y \notin \phi(VT_{(a,b)}(n;4)) $. If we successfully demonstrate this proposition, it enables us to infer, by contradiction, that for any sequences $X,Y \in \phi(VT_{(a,b)}(n;4))$ such that $D_1(X) \cap D_1(Y) \neq \emptyset$ and $w_X = w_Y$. \\

    Consider two binary sequences X and Y, each of length 2n such that $w_X \neq w_Y$, and $D_1(X) \cap D_1(Y) \neq \emptyset$. \\

    Let $\lambda$ and $\mu$ represent indices of $X$ and $Y$ from 1 to $2n$ respectively. Define sets $L$ and $M$ such that $L = \{\lambda\}$ and $M = \{\mu\}$. $X^L$ represents the sequence obtained after deleting all the bits indexed by $L$, and similarly, $Y^M$ represents the sequence obtained after deleting all the bits indexed by $M$.

    Given that $D_1(X) \cap D_1(Y) \neq \emptyset$, let's assume $\lambda$ and $\mu$ represent indices such that $X^L = Y^M$.

    From definition \ref{qary_Definition}, if $x \in VT_{a,b}(n, 4)$, then 
    \begin{equation}
        \sum_{i=1}^n x_i\;(mod\; q)  = b
    \end{equation}
    
    \begin{equation}
    \alpha(x)_i = 
    \left\{
    \begin{array}{ll}
      \text{1,} & \text{if } x_i \leq x_{i+1} \\
      \text{0,} & \text{otherwise}  \\
    \end{array}
    \right.
    \end{equation}
    We have provided a partial proof in three steps:

    \textbf{Step 1:} Find the difference of VT parameter ($b_x - b_y$).
    
    \textbf{Step 2:} Find the difference of VT parameters ($a_x - a_y$).
    
    \textbf{Step 3:} With exhaustive study of data, demonstrate that both ($a_x - a_y$) and ($b_x - b_y$) are not zero simultaneously. \\

    \textbf{Step1: Finding $b_x - b_y$ }\\ 
    If $b_x$ represents parameter $b$ of $x \in \Z_4^n$, then let $M_b(X)$ be a function of $X$, which we seek to determine that gives the value of $b$, where $X = \phi(x)$.

    Note: In this proof, we have continued discussion on $\phi_8$ from table \ref{tab:mapping}.

    \begin{align*}
        b_x &= M_b(X)\\
            &= \sum_{i=1}^nx_i\;(mod\; 4) 
    \end{align*}

    \textbf{From the Lemma \ref{lemma: conversion},}

    \begin{equation}
         M_b(X) = \sum_{i=1}^n3X_{2i-1} + X_{2i} - 2X_{2i-1}X_{2i}\;(mod\; 4)
    \end{equation}

    Likewise,

    \begin{align*}
        b_y &= M_b(Y)\\
            &= \sum_{i=1}^n3Y_{2i-1} + Y_{2i} - 2Y_{2i-1}Y_{2i}\;(mod\; 4) 
    \end{align*}

    Hence,
    \begin{align}
        b_x - b_y &= M_b(X) - M_b(Y)\;(mod\; 4) \nonumber \\
                  &=\sum_{i=1}^n3X_{2i-1} + X_{2i} - 2X_{2i-1}X_{2i} \nonumber \\
                  &\;-\sum_{i=1}^n3Y_{2i-1} + Y_{2i} - 2Y_{2i-1}Y_{2i}\;(mod\; 4) 
    \end{align}
    For $w_X \neq w_Y$, assume without sacrificing generality that, $w_X = w_Y + 1$. It holds because if $D_1(X) \cap D_1(Y) \neq \emptyset$, it implies that the difference between $w_X$ and $w_Y$ cannot exceed 1. Thus, considering both cases where $w_X - w_Y = 1$ and $w_Y - w_X = 1$ essentially cover all possibilities.
    
    Therefore, assumption $w_X = w_Y + 1$ is right without sacrificing generality. \\

    Now, as $X^L$ = $Y^M$ is true for $L = \{\lambda\}$ and $M = \{\mu \}$. If we assume $\lambda \leq \mu$, then we can write

    \begin{equation}
    Y_i = 
    \left\{
    \begin{array}{ll}
      X_i & \forall i < \lambda \; or \; i>\mu \\
      X_{i+1} & \forall \lambda < i < \mu  \\
    \end{array}
    \right.
    \end{equation}

    Also, as $w_X = w_Y + 1$ and $X^L$ = $Y^M$ $\implies$ $X_{i_k} = 1$ and $Y_{j_k} = 0$.

    Using the above properties, we can simplify equation 6 and reduce it to a function of X only. Hence we can represent $b_x - b_y$ in terms of X only.

    \begin{align}
    \label{equ_bx}
    \nonumber
        b_x - b_y &= 3 - 2X_{i_k+1} - 2X_{i_k+1}\{1 - X_{i_k+2}\}\\
        \nonumber
                  &+ \sum_{i = \frac{i_k+3}{2}}^{\frac{j_k - 1}{2}} 2{X_{2i-1} - X_{2i}} - 2X_{2i}\{X_{2i-1} - X_{2i-1}\}\\
                  &+ 2X_{j_k} - 2X_{j_k+1}\{X_{j_k}\}
    \end{align}

    \textbf{Step2: Finding $a_x - a_y$}

    As defined above,

    \begin{equation*}
    \alpha(x)_i = 
    \left\{
    \begin{array}{ll}
      \text{1,} & \text{if } x_i \leq x_{i+1} \\
      \text{0,} & \text{otherwise}  \\
    \end{array}
    \right.
    \end{equation*}

    \textbf{From lemma \ref{lemma_ax}, }we can rewrite the above equation in terms of X, where X = $\phi(x)$.

    \begin{align}
        \label{equ_alphax}
        \alpha(x)_i &= 1 - X_{2i-1} - X_{2i} + X_{2i}X_{2i+2} + X_{2i}X_{2i-1}\nonumber \\  
        &+ X_{2i-1}X_{2i+1} - X_{2i-1}X_{2i}X_{2i+2} - X_{2i}X_{2i+1}X_{2i+2}\nonumber \\
        &- X_{2i-1}X_{2i}X_{2i+1} - X_{2i-1}X_{2i+1}X_{2i+2}\nonumber \\ 
        &+ 2X_{2i-1}X_{2i}X_{2i+1}X_{2i+2}
    \end{align}

    Hence,
    \begin{equation}
        a_x = \sum_{i=1}^{n-1} i*\alpha(x)_i \;(\bmod \;n)
    \end{equation}

    Similarly,
    \begin{equation}
        a_y = \sum_{i=1}^{n-1} i*\alpha(y)_i \;(\bmod \;n)
    \end{equation}

    Therefore,
    \begin{equation}
        \label{equ_ax}
        a_x - a_y = \sum_{i=1}^{n-1} i*(\alpha(x)_i - \alpha(y)_i) \;(\bmod \;n)
    \end{equation}

    Using the facts, $w_X = w_Y + 1$ and $X^L = Y^M$ for L = $\{\lambda\}$ and M = $\{\mu\}$ and Also $X_{\lambda} = 1$ and $Y_{\mu} = 0$, we can get $a_x - a_y$ in terms of X only. \\

    \textbf{Step3: Exhaustive Study of Data} \\
    \par On exhaustive data analysis for all the possible pairs $x$ and $y$ till $n = 11$, equation \ref{equ_ax} and \ref{equ_bx} both are not zero simultaneously. We find the above statement open problem, as equation \ref{equ_ax} becomes too complex after using the equation \ref{equ_alphax}. Some examples are shown in table \ref{tab:ax-ay_data}. Data for all the possible pairs $x$ and $y$ till $n = 11$ is available at Github \cite{guptalab2024grayvt}. \\

    \begin{table}[h]
    \centering
    \caption{Data analysis of $a_X - a_Y$ and $b_X - b_Y$, where $X, \; Y \in \Z_2^{2n}$, $w_X = w_Y + 1$, and $D_1(X) \cap D_1(Y) \neq \emptyset$. Absolute differences are used to ensure focus on zero values of $a_X - a_Y$ and $b_X - b_Y$, preventing negative values from affecting the analysis.}

    \label{tab:ax-ay_data}
    \resizebox{\columnwidth}{!}{%
    \begin{tabular}{|c|c|c|c|c|}
    \hline
    \textbf{n} & \textbf{$X$} & \textbf{$Y$} & \textbf{$|a_X - a_Y|$} & \textbf{$|b_X - b_Y|$}\\
    \hline
    1 & 10 & 00 & 0 & 3 \\
    2 & 0001 & 0000 & 0 & 1 \\
    3 & 110001 & 100001 & 0 & 3 \\
    4 & 11100001 & 11000001 & 1 & 1 \\
    5 & 1011110111 & 1011101101 & 1 & 1 \\
    6 & 001000111001 & 001000011001 & 0 & 1 \\
    7 & 10111111100001 & 10011111100001 & 0 & 1 \\
    8 & 0001000011101000 & 0000100001101000 & 1 & 1 \\
    9 & 010111011010011101 & 010111010100101101 & 3 & 3 \\
    10 & 11110101011101111011 & 11110101011011110101 & 1 & 1 \\
    \hline
    \end{tabular}%
    }
    \end{table}

    Hence if $x \in VT_{a, b}(n, 4)$ then $y \notin VT_{a, b}(n, 4)$. Hence by method of contradiction, $\forall X,Y \in \phi(VT_{(a,b)}(n;4))$ and $D_1(X) \cap D_1(Y) \neq \emptyset$, then $w_X = w_Y$.

\begin{proposition}  
$\exists X,Y \in \phi(VT_{(a,b)}(n;4))$ such that residue of X = residue of Y and $D_1(X) \cap D_1(Y) \neq \emptyset$, where $\phi$ is any possible permutation out 24.
\end{proposition} 
\begin{proof}
   For example, consider X = (1,0,1,0,0,1,0,1) and Y = (1,0,1,0,0,0,1,1) for $\phi = \{0\rightarrow 00, 1\rightarrow 01, 2 \rightarrow 11, 3 \rightarrow 10\}$. The corresponding quaternary codewords of x and y are 3322 and 3302, respectively. From definition \ref{qary_Definition}, the residue of 3322 and 3302 is (0,0). Intersection of deletion spheres of x and y is $D_1(x) \cap D_2(y) = \{(1, 0, 1, 0, 0, 0, 1), (1, 0, 1, 0, 0, 1, 1)\}$. Detailed intersection results can be found on Github \cite{guptalab2024grayvt}.

\end{proof}
The results obtained from the study of Naisargik map over Varshamov-Tenengolts (VT) codes, provide valuable insights into the capabilities and limitations of these codes in handling error correction, particularly for single insertion or deletion errors. However, as the complexity of error increases, the effectiveness of VT codes may diminish. This presents an intriguing opportunity to explore the potential of Helberg codes, which have shown promise in addressing the limitations of VT codes by handling multiple insertion or deletion errors. This transition from VT codes to Helberg codes represents a natural progression in the pursuit of enhancing error correction capabilities, and the application of mapping techniques further amplifies the potential for advancements in this domain.
\section{Non-binary Helberg codes}
The Helberg code is an error-correcting code designed to handle multiple insertion or deletion errors during data transmission and storage. It builds upon the foundational principles of coding theory to offer a specialized approach to error correction, particularly in the context of DNA-based data storage. The exploration of Helberg codes holds immense promise in enhancing the fault tolerance of data storage systems, potentially revolutionizing the landscape of error correction mechanisms.

Consider a q-ary bit from $\Z_4$ and a codeword $x= (x_1,..., x_n)$ of length $n$ that belongs to $\Z_q^n$. Each symbol $\x_i$ in $x$ is referred as the $i^{th}$ symbol. Establish $s$ as a positive integer that signifies the maximum permissible deletions in any codeword, and set $r = q - 1$.

For The sequence of weights $v(q, s) = \{v_1(q, s), v_2(q, s), ...\}$. The weight $v_i(q, s)$ is defined recursively as follows:
\begin{equation} \label{eq:1}
v_i(q, s) = 1  + r \sum_{j=1}^{s} v_{i-j}(q, s),
\end{equation}
With $v_i(q, s) = 0$ for $i \leq 0$. From this point forward, we will use $v_i$ as a shorthand for $v_i(q, s)$.
\\
then define the moment of $\Bar{x}$ as
\begin{equation}
    \label{eq:def_M(x)}
    M(x)= \sum_{i=1}^n v_i x_i.
\end{equation}
The truncated codeword $(x)_k = (x_1, \ldots, x_k)$ consists of the first $k$ symbols of $x$. Its moment $M_k(x)$ can also be denoted by $M((x)_k)$. The moment of the symbol $x_i$ is denoted as $M(x_i) = w_i x_i$. These $q$-ary codes can correct multiple insertion and deletion errors.\cite{le2023new}

\begin{definition}
\label{def_5}
The Generalized Helberg code can thus be defined as 
\[
H(n,q,s,a) =\{ (x_1, x_2, \ldots, x_n) \in 
\Z_q^n:
\]
\begin{equation}
\label{helbergdefinition}
\sum_{i=1}^n v_i x_i \equiv a\pmod{m}\, 
\end{equation}
\begin{equation*}
 \therefore  M(x) \equiv a\pmod{m}\,\}
\end{equation*}
where,
\begin{equation*}
v_i = \begin{cases}
0, & \text{for } i \leq 0 \\
1 + r \sum_{j=1}^{s} v_{i-j}(q, s), & \text{for } i > 0
\end{cases}.
\end{equation*}

\begin{equation}
\label{m}
m = r \sum_{i=0}^{s-1} v_{n-i} + 1,
\end{equation}
\end{definition}
In the context of our study, $m$ is basically $v_{n+1}$, where $v_i$ represents a sequence of non-negative integers that are monotonically increasing for all $i > 0$. It is important to note that both $v_i$ and $m$ are functions of $s$.
Furthermore, we introduce an integer $a$, which we refer to as the 'residue'. The residue is associated with a specific moment defined by $m$. 

Lastly, we denote $H(n, q,s, a)$ as a Helberg code of length $n$. This notation is in accordance with the work of Abdelghaffar (2012)\cite{abdelghaffar2012}.

Readers may refer to \cite{abdelghaffar2012,le2023new} for a more comprehensive understanding. Here, is the small example of generating Helberg code with desired parameters.

\begin{example}
Consider parameters  n = 5 ,q = 4 ,s = 3 and a = 0.0. First generate the v from equation \ref{eq:1} and m as $v_{n+1}$.
\begin{align*}
    v = [1, 4, 16, 61, 232] \\
    m = v_{n+1} = 880 
\end{align*}
Put all the values in equation \ref{helbergdefinition} which results in $H(5, 4, 3, 0.0)$. The corresponding codewords are $\{(0, 0, 0, 0, 0), (1, 0, 0, 3, 3), (2, 3, 3, 2, 3)\}$.
\end{example}
Helberg code itself is very powerful optimal code. Naisargik map may enhance the scope of error correcting capability of Helberg code by changin the domain. \\
\textbf{Note}: This section refers Naisargik map $\phi_9$ of Table \ref{tab:mapping} as $\phi$.
\subsection{\textbf{Naisargik Image of Quaternary Helberg Code}}

Consider codeword $x \in \Z_4^n$. Map x with the help of $\phi$ such that $\phi(x) \in \Z_2^{2n}$. Generalized Helberg code $H(n,4,s, a)$ is the subset of $\Z_4^n$ and Naisargik image of generalized Helberg code $\phi(H(n,4,s, a))$ is the subset of $\Z_2^{2n}$. Mathametically,
$$H(n,4,s, a) \in \Z_4^n \implies \phi(H(n,4,s, a)) \in \Z_2^{2n}$$

The observations states that Naisargik map enhance the capability of deletion error correction means if $H(n,4,s, a)$ is s deletion error correcting code, then $\phi(H(n,4,s, a))$ is (s+1) deletion error correcting code. Through thorough testing, we discovered that this method could create codes that can correct deletion errors for lengths up to 12 and for up to 5 deletions. Some of these results are given in Table \ref{tab:my_tabul}. The mathematical proof of the observation is discussed further. The example of the mapping is given below.

\begin{example}
Suppose $n = 4$, $q = 4$, and $s = 1$. Calculate $v =  [ 1.0,  4.0, 13.0, 40.0]$ and $m = 121.0$ with the help of equation \ref{eq:1} and $v_{n+1}$ respectively. Genrate all the residues and corresponding codewords from definition \ref{def_5}. Table \ref{tab:residue_values} shows the number of codewords for each residue value. residue values 40.0 and 13.0 gives the maximum number of codewords which is 5. 

Consider residue value a = 13.0 and apply the Naisargik map $\phi$ over the codewords of a = 13.0. Table \ref{tab:quathelbergexample} show the mapping of $H(4,4,1,13.0)$ to $\phi(H(4,4,1,13.0))$. The Quaternary Helberg codeword of length n changes to length 2n.

\begin{table}[h]
\centering
\caption{Naisargik mapping of $H(4,4,1,13.0)$ to $\phi(H(4,4,1,13.0))$}
\label{tab:quathelbergexample}
\begin{tabular}{|c|c|}
\hline
\textbf{Original Quaternary} & \textbf{Naisargik Mapped} \\
\textbf{Codewords} & \textbf{Binary Codewords} \\
\hline
(0, 0, 1, 0) & (1, 1, 1, 1, 0, 1, 1, 1) \\
\hline
(1, 0, 1, 3) & (0, 1, 1, 1, 0, 1, 0, 0) \\
\hline
(1, 3, 0, 0) & (0, 1, 0, 0, 1, 1, 1, 1) \\
\hline
(2, 3, 0, 3) & (1, 0, 0, 0, 1, 1, 0, 0) \\
\hline
(3, 3, 3, 2) & (0, 0, 0, 0, 0, 0, 1, 0) \\
\hline
\end{tabular}
\end{table}

Decode $\phi(H(4,4,1,13.0))$ using deletion sphere approach shown in  Algorithm \ref{sphere}. Table \ref{tab:deletion_sphere} shows the deletion sphere corresponding to each codeword. From the table it is easy to observe that if senders sends the codeword $ (0, 1, 1, 1, 0, 1, 0, 0)$ and any two random bit deletion occur, then receiver can identify the original codeword as deletion sphere do not intersect with any other codeword within residue a = 13.0.\\

By observing more such cases we came to conclusion that if $H(n,4,s,a)$ is s deletion error correcting code then $\phi(H(n,4,s,a))$ is (s+1) deletion error correcting code. To prove the observation mathematically, following Lemmas are useful. Theorem \ref{thm2} is the actual mathematical proof of the observation.
\end{example}

In the following lemmas, Lemma \ref{lemma_1} and Lemma \ref{lemma_2}, important properties of coefficient $C$ are demonstrated, which are helpful in proving Theorem \ref{thm2}.

The expression $C_i = ((i+1)\ (\bmod\ 2) + 1)v_{\lceil \frac{i}{2} \rceil}$ represents the coefficient at the $i^{th}$ position.

\begin{lemma}
    \label{lemma_1}
    $\forall i,j \in \n$, if $i > j$ then $C_i > C_j$.

\end{lemma}
\begin{proof}
    \textbf{(Outline)} We will prove this statement in two cases to cover the entire range of $i$.
    \begin{itemize}
        \item \textbf{Case 1:} $i$ is even.
        \item \textbf{Case 2:} $i$ is odd. \\
    \end{itemize}

    \textbf{Case 1 (i is even):} $C_i = 2v_\frac{i}{2}$, and $C_{i-1} = v_\frac{i}{2}$. \\ \\Now for $i>0 => v_i > 0$ from \cite{helberg}.

$$\therefore 2v_i > v_i$$.
$$\therefore C_i > C_{i-1}$$.
Therefore, the given lemma is true in this case.\\
\textbf{Case 2 (i is odd):} $C_{i} = v_\frac{i+1}{2}$, and $C_{i-1} = 2v_\frac{i-1}{2}$.\\
Now, $v_i = 1 + 3\sum_{j = 1}^dv_{i-j}$. \\
As $s \geq 1$,
    $$v_i >= 1 + 3v_{i-1}$$.
    $$\therefore v_i > 3v_{i-1}$$
    $$\therefore C_i > C_{i-1}$$
    
Therefore, the given lemma is also true in this case.

As cases 1 and 2 generalize the range of i, the given statement is true for all $\forall i \in \n$.

Therefore we can say that, $\forall i>j$, $C_i > C_j$. Where $i, j > 0$ and $C_i = ((i+1)(\bmod2) + 1)v_{\lceil \frac{i}{2} \rceil}$.
\end{proof}

\begin{lemma}
    \label{lemma_2}
    $$
        C_L - \sum_{i=L-s}^{L-1} C_i \geq 1,
    $$
    where $i, \;L \in \n$.
\end{lemma}
\begin{proof}
    \textbf{(Outline)} We will prove this statement in two cases to cover the entire range of $L$.
    \begin{itemize}
        \item \textbf{Case 1:} $L$ is even.
        \item \textbf{Case 2:} $L$ is odd. \\
    \end{itemize}
    \textbf{Case 1: Let L be even} 
    $$C_L = 2v_{\frac{L}{2}}\;and\; C_{L-1}=v_{\frac{L}{2}}$$.
    \begin{align*}
        v_{\frac{L}{2}} &= 1 + 3\sum_{i=1}^dv_{\frac{L}{2}-i}\\
                        &= 1 + 3v_{(\frac{L}{2}-1)} + 3v_{(\frac{L}{2}-2)}+...+3v_{(\frac{L}{2}-s)}\\
                        &= 1 + C_{L-2} + C_{L-3} +...+ C_{L-2s} + C_{L-2s-1}\\
                        &= 1 + \sum_{i=L-2s-1}^{L-2}C_i
    \end{align*}
    Now,
    \begin{align*}
        C_L - \sum_{i=L-s}^{L-1} C_i &= C_L - C_{L-1} - \sum_{i=L-s}^{L-2}C_i\\
                                    &= 2v_{\frac{L}{2}} - v_{\frac{L}{2}} - \sum_{i=L-s}^{L-2}C_i\\
                                    &= 2v_{\frac{L}{2}} - v_{\frac{L}{2}} - \sum_{i=L-s}^{L-2}C_i\\
                                    &= v_{\frac{L}{2}} - \sum_{i=L-s}^{L-2}C_i\\
                                    &= 1 + \sum_{i=L-2s-1}^{L-2}C_i - \sum_{i=L-s}^{L-2}C_i
    \end{align*}    
    Consider,
    \begin{align*}
        L-2s-1 &= L-s-(s+1)\\
                &< L-s \;\;\;\;\;\;\;\;(\because s\geq 1)
    \end{align*}
    So, we can say that $\forall i, C_i\geq 0$. So,
    $$\sum_{i=L-2s-1}^{L-2}C_i \geq \sum_{i=L-s}^{L-2}C_i$$
    \begin{equation}
        \label{case 1}
        \therefore C_L - \sum_{i=L-s}^{L-1} C_i \geq 1
    \end{equation}
    
    \textbf{Case 2: If L is odd}
    $$C_L = v_{\frac{L+1}{2}}$$
    Now, 
    \begin{align*}
        v_{\frac{L+1}{2}} &= 1 + 3\sum_{i=1}^dv_{\frac{L+1}{2}-i}\\
                        &= 1 + 3v_{\frac{L-1}{2}}+3v_{\frac{L-3}{2}}+...+3v_{\frac{L+1-2s}{2}}\\
                        &= 1 + C_{L-1} + C_{L-2} + C_{L-3} +...+ C_{L+1-2s} + C_{L-2s}\\
                        &= 1 + \sum_{i=L-2s}^{L-1}C_i
    \end{align*}
    \begin{align*}
        \therefore C_L - \sum_{i=L-s}^{L-1} C_i &= v_{\frac{L+1}{2}} - \sum_{i=L-s}^{L-1} C_i\\
        &= 1 + \sum_{i=L-2s}^{L-1}C_i - \sum_{i=L-s}^{L-1} C_i
    \end{align*}
    As $s \geq 1$, 
    $$\sum_{i=L-2s}^{L-1}C_i \geq \sum_{i=L-s}^{L-1} C_i\;(\because L-2s \leq L-s)$$.
    
    \begin{equation}
        \label{case2}
        \therefore C_L - \sum_{i=L-s}^{L-1} C_i \geq 1
    \end{equation}
    Equation \ref{case 1} and \ref{case2} proves that for any positive L, our Lemma holds true.
    \begin{align}
        \therefore C_L - \sum_{i=L-s}^{L-1} C_i \geq 1
    \end{align}
\end{proof}

\begin{assumption}
\label{helberassumption}
    The Naisargik map $\phi$ is refered as $\phi_9$ from the Table \ref{tab:mapping} which is $\phi:\{0\rightarrow11,1\rightarrow01,2\rightarrow10,3\rightarrow00\}$.
\end{assumption}

From Table \ref{tab:quathelbergexample} and under the given assumption \ref{helberassumption}, we can establish relations between $x$ and $X = \phi(x)$, where $x \in \Z_4^n$ and $X \in \Z_2^{2n}$. In the following lemma \ref{lemma_3}, we formally state this relation.

\begin{lemma}
    \label{lemma_3}
    For any $X \in \Z_2^{2n}$, if we say $X = \phi(x)$ where $x \in \Z_4^{n}$, then
    \[
        x_i = 3 - X_{2i-1} - 2X_{2i}
    \]
    Here, $x_i$ represents the $i$-th bit of $x$, and similarly, $X_i$ represents the $i$-th bit of $X$, where $1 \leq i \leq n$.
\end{lemma}

\begin{proof}
    Let's assume the following linear relation for a specific bit $i$:
    \[
        ax_i = bX_{2i-1} + cX_{2i} + d
    \]
    where $a$, $b$, $c$, and $d$ are constants to be determined.

    Now, we get four equations for all four inputs:
    \begin{align*}
        \text{For } X_{2i-1} = 0, \; X_{2i} = 0, \; and \; x_i = 3: & \quad a \cdot 3 = d \\
        \text{For } X_{2i-1} = 0, \; X_{2i} = 1, \; and \; x_i = 1: & \quad a \cdot 1 = b - c + d \\
        \text{For } X_{2i-1} = 1, \; X_{2i} = 0, \; and \; x_i = 2: & \quad a \cdot 2 = -b + d \\
        \text{For } X_{2i-1} = 1, \; X_{2i} = 1, \; and \; x_i = 0: & \quad a \cdot 0 = -b - c + d 
    \end{align*}
    
    Solving these four equations for four variables, we obtain $a = 1$, $b = -1$, $c = -2$, and $d = 3$. Therefore, the assumed relation holds true for all possible cases, establishing the validity of the lemma.
\end{proof}

\begin{example}
Let's verify Lemma \ref{lemma_3} from an example:\\
Let $X = 10001100$,
According to Lemma \ref{lemma_3}:
\[
    x_i = 3 - X_{2i-1} - 2X_{2i}
\]
For $i = 1$: $x_1 = 3 - 1 - 2 \times 0 = 2$\\
For $i = 2$: $x_2 = 3 - 0 - 2 \times 0 = 3$\\
For $i = 3$: $x_3 = 3 - 1 - 2 \times 1 = 0$\\
For $i = 4$: $x_4 = 3 - 0 - 2 \times 0 = 3$\\ \\
Hence, the 4-bit quaternary sequence $x$ corresponding to $X = 10001100$ is $x = 2303$. We can also verify this from table \ref{tab:quathelbergexample}.
\end{example}

\begin{theorem}  
\label{thm2}
If H(n, 4, s, a) is s error correcting code then $\phi(H(n, 4, s, a))$ is $(s+1)$ deletion error-correcting binary code. Where, $\phi(H(n, 4, s, a))$ is Naisargik image of quaternary Helberg code.

As per research conducted by Levenshtein \cite{Levenshtein1965BinaryCC}, $\phi(H(n,4,s, a))$ can correct up to $(s+1)$ indels.
\end{theorem}
\begin{proof}
Let $x, y \in \Z_4^n$. If $X = \phi(x)$ and $Y = \phi(y)$, then it follows that $X, Y \in \Z_2^{2n}$.

Let $D, E \subseteq \{1, 2, \ldots, 2n\}$ such that $|D| = |E| = s + 1$, and let $X^D = (X_{i_1}, \ldots, X_{i_{n'}})$ where $n' = 2n - |D|$. If we delete all the bits from $X$ indexed by $D$, we get $X^D$ and all the bits from $Y$ indexed by $E$, we get $Y^E$, and these sequences will be the same, i.e., $X^D = Y^E$. \\

Suppose, $\Delta (x,y) = M(x) - M(y)$, where, $ M(x) = \sum_{i=1}^n v_i x_i$ and $v_i$ is the $i^{th}$ Helberg coefficient given in defination \ref{def_5}.\\

\begin{equation}
    \label{maincond}
    0 < |\Delta (x,y)| < m
\end{equation}

Here $m = v_{n+1}$ as defined in defination \ref{def_5}. If the above condition is satisfied then x and y belong to different Helberg codes \cite{le2023new}.

From Lemma \ref{lemma_3} results, we can replace $x_i$ as following
$$x_i = 3-X_{2i-1}-2X_{2i}$$

\begin{align*}
    Consider,\;M(x) &= \sum_{i=1}^n v_i x_i\\
         &= \sum_{i=1}^n v_i(3-X_{2i-1}-2X_{2i})\\
         &= \sum_{i=1}^n 3v_i - \sum_{i=1}^{2n}((i+1) (mod\;2) + 1)v_{\lceil \frac{i}{2} \rceil}(X_i)
\end{align*}
Similarly,
$$M(y) = \sum_{j=1}^n 3v_j - \sum_{j=1}^{2n}((j+1)(\bmod\;2) + 1)v_{\lceil \frac{j}{2} \rceil}(Y_j)$$
Hence,
\begin{align*}
    \Delta (x,y) &= M(x) - M(y)\\ 
    &= \sum_{i=1}^n 3v_i - \sum_{i=1}^{2n}((i+1)(\bmod\; 2) + 1)v_{\lceil \frac{i}{2} \rceil}(X_i)\\ 
    &\;\;\;\;- \sum_{j=1}^n 3v_j + \sum_{j=1}^{2n}((j+1)(\bmod\; 2) + 1)v_{\lceil \frac{j}{2} \rceil}(Y_j)\\
    &= \sum_{j=1}^{2n}((j+1)(\bmod\; 2) + 1)v_{\lceil \frac{j}{2} \rceil}(Y_j)\\
    &\;\;\;\;- \sum_{i=1}^{2n}((i+1)(\bmod\; 2) + 1)v_{\lceil \frac{i}{2} \rceil}(X_i)
\end{align*}
Write $C_i = ((i+1)(\bmod\;2) + 1)v_{\lceil \frac{i}{2} \rceil}$.
we can write $\Delta (x,y)$ as the difference between different bits indexed by D and E with common bits in both.

\begin{align}
    \label{main_equation_equal}
    \therefore \Delta (x,y) &= \sum_{j\in E} C_jY_j - \sum_{i\in D} C_iX_i + \sum_{k=1}^{n'} (C_{j_k} - C_{i_k})X_{i_k}
\end{align}

Where $k \in S(n')$ and $S(n')$ = \{1, 2, ..., 2n\}. $i_k$ and $j_k$ replicates those indices whose bit values are equal in X and Y i.e. $X_{i_k} = Y_{j_k}\; \forall k \in S(n')$.

\textbf{Outline:} Now, we will prove $0 < \Delta (x,y) < m$ in two steps. In the first step, we will prove that the upper bound as $m$, and in the next step, we will prove the lower bound.

\begin{itemize}
    \item \textbf{Step 1:} Prove the upper bound as $m$.
    \item \textbf{Step 2:} Prove the lower bound as 0. \\
\end{itemize}

\textbf{Step 1:} $\Delta (x,y) < m$
\begin{align*}
    \Delta (x, y) &\leq \sum_{j \in E}C_jY_j + \sum_{k=1}^{n'}(C_{j_k}-C{i_k})X_{i_k}
\end{align*}

We proceed by dividing $S(n') = \{1, 2, \ldots, n'\}$ into two subsets: one containing elements $k$ where $j_k \leq i_k$, and the other containing elements where $j_k > i_k$.

\begin{align*}
    \Delta (x,y) &\leq \sum_{j \in E}C_jY_j + \sum_{\substack{k\in S(n')\\ j_k>i_k}}(C_{j_k}-C{i_k})X_{i_k}\\
                &+ \sum_{\substack{k\in S(n')\\ j_k\leq i_k}}(C_{j_k}-C_{i_k})X_{i_k}
\end{align*}
From Lemma \ref{lemma_1}, if $i_k>j_k$, then $C_{i_k}>C_{j_k}$ $\forall i_k, j_k \in \n$.

\begin{align*}
\therefore \sum_{\substack{k\in S(n')\\ j_k\leq i_k}}(C_{j_k}-C_{i_k})X_{i_k} &\leq 0\\ \\
\therefore \Delta (x,y) \leq \sum_{j \in E}C_jY_j + \sum_{\substack{k\in S(n')\\ j_k>i_k}}&(C_{j_k}-C_{i_k})X_{i_k}
\end{align*}

For maximum value, putting $Y_j = X_{i_k} = 1$. As Y and X are binary strings.

\begin{align*}
        \therefore \Delta (X,Y) &\leq \sum_{j \in E}C_j + \sum_{\substack{k\in S(n')\\ j_k>i_k}}(C_{j_k}-C_{i_k})\\
                &= \sum_{j\in E} C_j + \sum_{\substack{k\in S(n')\\ j_k>i_k}} C_{j_k} - \sum_{\substack{k\in S(n')\\ j_k>i_k}} C_{i_k}\\
                &\;\;\;\;+ \sum_{\substack{k\in S(n')\\ j_k\leq i_k}} C_{j_k} -  \sum_{\substack{k\in S(n')\\ j_k\leq i_k}} C_{j_k}\\
                &= \sum_{j=1}^{2n}C_j - \sum_{k\in S(n')}C_{min(i_k,j_k)}\\
                &\leq \sum_{j=1}^{2n}C_j - \sum_{k=1}^{n'}C_k \;\;\;\;\;\;(\because n'\; terms)\\
                &\leq \sum_{i=n'+1}^{2n}C_i\\
                &\leq \sum_{i=2n-s}^{2n}C_i\\
                &= \sum_{i=1}^{s+1} C_{2n+1-i}
\end{align*}
\begin{equation}
    \label{eq:13}
   \therefore \Delta (x,y) \leq \sum_{i=1}^{s+1} C_{2n+1-i}
\end{equation}
Now, we know that $C_i = ((i+1)\%2 + 1)v_{\lceil \frac{i}{2} \rceil}$.
\[
    C_i= 
\begin{cases}
    2v_{\lceil \frac{i}{2} \rceil},& \text{if i is even}\\
    v_{\lceil \frac{i}{2} \rceil},              & \text{if i is odd}
\end{cases}
\]
\begin{align*}
    3v_i &= C_{2i-1} + C_{2i}\\
    3\sum_{j=1}^dv_{n+1-j} &=\sum_{i=1}^{2s}C_{2n+1-i}\\
    m &= 1 + \sum_{i=1}^{2s}C_{2n+1-i}
\end{align*}

Now, $s+1\leq 2s$, $\forall s \in \n$,
\begin{align*}
    \sum_{i=1}^{2s}C_{2n+1-i} \geq \sum_{i=1}^{s+1}C_{2n+1-i}\\
\end{align*}
\begin{equation}
    \label{eq:14}
    \therefore m>\sum_{i=1}^{s+1}C_{2n+1-i}
\end{equation}

Therefore, from equations \ref{eq:13} and \ref{eq:14}
$$ 
    \Delta (x,y) < m
$$

\textbf{Step 2:} $\Delta (x,y) > 0$\\ \\
Without loss of generality, we may assume $L \in \{1,2,..,2n\}$, such that $Y_L > X_L$ and $X_i = Y_i = 0\; \forall i>L$ as proved in \cite{le2023new}. Now, $Y_L = 1$ and $X_L = 0$ because X and Y are in binary.\\

Position of $L$ with respect to sets $D$ and $E$, we can break this problem into four different cases as follows:

\begin{itemize}
    \item \textbf{Case I:} $L \in D \cap E$.
    \item \textbf{Case II:} $L \in E - D$.
    \item \textbf{Case III:} $L \in D - E$.
    \item \textbf{Case IV:} $L \notin D \cup E$. \\
\end{itemize}

\textbf{Case I: $L \in D \cap E$.}\\ \\
From, equation \ref{main_equation_equal}, \\
\begin{align*}
    \therefore \Delta (x,y) &= \sum_{j\in E} C_jY_j - \sum_{i\in D} C_iX_i + \sum_{k=1}^{n'} (C_{j_k} - C_{i_k})X_{i_k}
\end{align*}

On observing the first two terms of the above equation, and putting $Y_L$ and $X_L$ values.

\begin{align*}
    \sum_{j\in E} C_jY_j - \sum_{i\in D} C_iX_i &= C_L + \sum_{\substack{j\in E\\j \leq L-1}}C_jY_j - \sum_{\substack{i\in D\\i \leq L-1}}C_iX_i\\
    &\geq C_L - \sum_{\substack{i\in D\\i \leq L-1}}C_iX_i
\end{align*}
And the last term $\sum_{k=1}^{n'} (C_{j_k} - C_{i_k})X_{i_k}$ can be written in the form of addition of positive and negative terms because of lemma \ref{lemma_1}.
\begin{align*}
    \sum_{k=1}^{n'} (C_{j_k} - C_{i_k})X_{i_k} &= \sum_{\substack{k \in S(n')\\ j_k \geq i_k}} (C_{j_k} - C_{i_k})X_{i_k} \\
    &\;\;\;\;+ \sum_{\substack{k \in S(n')\\ j_k < i_k}} (C_{j_k} - C_{i_k})X_{i_k}\\
    & \geq \sum_{\substack{k \in S(n')\\ j_k < i_k \\i_k \leq L-1}} (C_{j_k} - C_{i_k})X_{i_k}
\end{align*}
Combining all the terms,
\begin{equation}
    \label{mainequation}
    \Delta (x,y) \geq C_L - \sum_{\substack{i\in D\\i \leq L-1}}C_iX_i + \sum_{\substack{k \in S(n')\\ j_k < i_k \\i_k \leq L-1}} (C_{j_k} - C_{i_k})X_{i_k}
\end{equation}
For minimum value putting $X_i = X_{i_k} = 1$,
\begin{align*}
    \Delta (x,y) &\geq C_L - \sum_{\substack{i\in D \\ i\leq L-1}}C_{i} - \sum_{\substack{k\in S(n')\\j_k<i_k \\ i_k\leq L-1}}C_{i_k} - \sum_{\substack{k\in S(n')\\j_k\geq i_k \\ i_k\leq L-1}}C_{i_k} \\
    &\;\;\;+ \sum_{\substack{k\in S(n')\\j_k<i_k \\ i_k\leq L-1}}C_{j_k} + \sum_{\substack{k\in S(n')\\j_k\geq i_k \\ i_k\leq L-1}}C_{i_k}\\
    &= C_L - \sum_{i=1}^{L-1}C_i + \sum_{\substack{k\in S(n')\\ i_k\leq L-1}} C_{min(i_k,j_k)} \\
    &\Biggl (\because \sum_{\substack{i\in D \\ i\leq L-1}}C_{i} + \sum_{\substack{k\in S(n')\\j_k<i_k \\ i_k\leq L-1}}C_{i_k} + \sum_{\substack{k\in S(n')\\j_k\geq i_k \\ i_k\leq L-1}}C_{i_k} = \sum_{i=1}^{L-1}C_i \Biggr ) 
\end{align*}

From Lemma \ref{lemma_1}, $C_{i_k} > C{j_k}$, if $i_k > j_k$
\begin{align*}
    \therefore \Delta (x,y) &\geq C_L - \sum_{i=1}^{L-1}C_i + \sum_{k=1}^{min(n',L-1)}C_k\\
                &\geq C_L - \sum_{i=min(n',L-1)+1}^{L-1}C_i\\
                &\geq C_L - \sum_{i=L-s}^{L-1} C_i
\end{align*}
Consider $L\leq 2n = n' + s+1 \;and\; n' \geq L-(s+1)$.
\begin{equation}
    \Delta (x,y) \geq C_L - \sum_{L-s}^{L-1} C_i
\end{equation}

From lemma \ref{lemma_2}, $C_L - \sum_{i=L-s}^{L-1} C_i \geq 1$,\\
\begin{align}
    \therefore \Delta(x,y) \geq 1 \implies \Delta(x,y) > 0
\end{align}

\textbf{Case II: $L \in E-D$.}\\
From, equation \ref{main_equation_equal}, \\
\begin{align*}
    \Delta(x,y) &= \sum_{j\in E} C_jY_j - \sum_{i\in D} C_iX_i + \sum_{k=1}^{n'} (C_{j_k} - C_{i_k})X_{i_k}\\
    &=C_L - \sum_{\substack{i \in D \\ i \leq L-1}}C_iX_i + \sum_{\substack{k \in S(n')\\ j_k < i_k \\i_k \leq L-1}} (C_{j_k} - C_{i_k})X_{i_k}\\
    &\;\;\;\;+\sum_{\substack{j \in E \\ j \leq L-1}}C_jY_j + \sum_{\substack{k \in S(n')\\ j_k \geq i_k \\i_k \leq L-1}} (C_{j_k} - C_{i_k})X_{i_k}\\
\end{align*}
So,
$$
    \Delta (x,y) \geq C_L - \sum_{\substack{i\in D\\i \leq L-1}}C_iX_i + \sum_{\substack{k \in S(n')\\ j_k < i_k \\i_k \leq L-1}} (C_{j_k} - C_{i_k})X_{i_k}
$$

Above equation is same as equation \ref{mainequation} received in case 1. Therefore, by same approach, for this case,

$$\Delta(x,y) > 0$$

\textbf{Case III: $L \in D-E$}
It is easy to observe that if we swap the D and E in case ii, then it will become case iii. For case iii also,

$$\Delta(x,y) > 0$$

\textbf{Case IV: $L \notin D\cup E$}\\
Consider $\exists\; K \in S(n'),\;j_K = L$ such that $ Y_L > X_L \implies Y_{j_K}>X_{j_K}$ and $X_i = Y_i = 0,\;\forall i > L$. $i_K \leq j_K - 1$ such that, $X_{i_K} = Y_{j_K}$.

From equation \ref{main_equation_equal}, 
$$\Delta (x,y) = C_L - \sum_{\substack{i\in D\\i \leq L-1}}C_iX_i + \sum_{\substack{k \in S(n')\\ j_k < i_k \\i_k \leq L-1}} (C_{j_k} - C_{i_k})X_{i_k}$$
Now $\forall k \in S(n'),\;X_{i_k}=Y_{j_k}$,
\begin{align*}
    \Delta (x,y) &\geq - \sum_{\substack{i\in D\\i \leq L-1}}C_iX_i + \sum_{k\in S(n')}(C_{j_k} - C_{i_k})Y_{j_k}\\
    &= - \sum_{\substack{i\in D\\i \leq L-1}}C_iX_i + \sum_{\substack{k\in S(n')\\ j_k <i_K}}(C_{j_k} - C_{i_k})Y_{j_k}\\
    &\;\;\;\;+ \sum_{\substack{k\in S(n')\\ j_k \geq i_K}}(C_{j_k} - C_{i_k})Y_{j_k}\\
    &\geq - \sum_{\substack{i\in D\\i \leq L-1}}C_iX_i + \sum_{\substack{k\in S(n')\\ j_k <i_K}}(C_{j_k} - C_{i_k})Y_{j_k}\\
    &\;\;\;\;+(C_L - C_{i_k})Y_L
\end{align*}
Put $Y_L = 1, X_{i_k} = Y_{j_k}$
\begin{equation}
    \Delta (x,y) \geq C_L - \sum_{\substack{i\in D\\i \leq L-1}}C_iX_i + \sum_{\substack{k \in S(n')\\ j_k < i_k \\i_k \leq L-1}} (C_{j_k} - C_{i_k})X_{i_k} - C_{i_k}
\end{equation}

Above equation is similar to equation \ref{mainequation}. So, it follows the case i. approach. Hence, for this case,

$$\Delta(x,y) > 0$$

As these four cases covers generality, till this step we have proved $0 < \Delta(x,y) < m$.

It is easy to prove $0<\Delta(y,x)<m$ by following the same process and only reversing y and x's roles. By combining all the results,

\begin{equation*}
    0<|\Delta(x,y)|<m
\end{equation*}

Since, in this case above condition is satisfied x and y must belong to different Hlberg Codes. In getting of this result, out initial assumption was $X^D = Y^E$.

Thus if $X^D = Y^E$ is true, then x and y must be from different Helberg codes. And since $X^D$ is at $(s+1)$ deletion distance from X, $\phi(H(n, 4, s, a))$ becomes $(s+1)$ deletion correction code.

As per research conducted by Levenshtein \cite{Levenshtein1965BinaryCC}, $\phi(H(n,4,s, a))$ can correct up to $(s+1)$ indels.
\end{proof}
Hence, Naisargik map can enhance the deletion error correcting capability of Helbercode. Likewise inverse Naisargik map can also extend the scope of deletion error correction for Helberg code. Next subsection discusses about the influence of inverse Naisargik map over binary Helberg code.
\subsection{\textbf{Inverse Naisargik Image Of Binary Helberg Code}}
Inverse Naisargik map will map the codeword as $\Z_2^{2n} \rightarrow \Z_4^n$. Consider $x \in \Z_2^{2n}$. Inverse Naisargik map $\phi^{-1}$ will map x to $\phi^{-1}(x) \in \Z_4^n$. All the binary Helberg codewords of length 2n can be generated from definition \ref{def_5}. H(2n,2,s,a) contains possible residue values a and set of codewords corresponding to a. $\phi^{-1}$ will map H(2n,2,s,a) to $\phi^{-1}(H(2n,2,s, a))$. Mathematically,
$$H(2n,2,s, a) \in \Z_2^{2n} \rightarrow \phi^{-1}(H(2n,2,s, a)) \in \Z_4^n$$

The example below takes the feasible values of the required parametrs and briefs about the process of mapping the Helberg code with Naisargik Helberg code.

\begin{example}
Suppose the Helberg parameters are n = 10, q = 2 and s = 2. Generate all the residues and corresponding codeword set with equation \ref{eq:1}. Calculate the values $v=[ 1, 2, 4, 7, 12, 20, 33, 54, 88, 143]$ and m = 232.0 from the equation \ref{eq4} and \ref{m} respectively. Table \ref{tab:residue_values_rev} shows the number of codewords for each residue value. It is easy to observe that residue value 66.0 gives the maximum number of codewords. 

For residue value a = 66.0, apply inverse Naisargik map $\phi^{-1}$ on corresponding codewords and generate the images. Table \ref{tab:my_label_rev} shows the inverse Naisargik images of H(10, 2, 2, 66.0).

\begin{table}[h]
\centering
\caption{Inverse Naisargik mapping of H(10, 2, 2, 66.0)}
\label{tab:my_label_rev}
\begin{tabular}{|c|c|}
\hline
\textbf{Codewords of H(10, 2, 2, 66.0)} & \textbf{Inverse Naisargik image} \\
\hline
(0, 0, 0, 0, 1, 0, 0, 1, 0, 0) & (3, 3, 2, 1, 3) \\
\hline
(0, 1, 0, 0, 1, 1, 1, 0, 1, 1) & (1, 3, 0, 2, 0) \\
\hline
(0, 1, 1, 1, 0, 0, 0, 1, 1, 1) & (1, 0, 3, 1, 0) \\
\hline
(0, 1, 1, 1, 0, 1, 1, 0, 0, 0) & (1, 0, 1, 2, 3) \\
\hline
(1, 0, 0, 0, 1, 0, 0, 1, 1, 1) & (2, 3, 2, 1, 0) \\
\hline
(1, 0, 0, 0, 1, 1, 1, 0, 0, 0) & (2, 3, 0, 2, 3) \\
\hline
(1, 0, 1, 1, 0, 0, 0, 1, 0, 0) & (2, 0, 3, 1, 3) \\
\hline
(1, 1, 1, 1, 0, 1, 1, 0, 1, 1) & (0, 0, 1, 2, 0) \\
\hline
\end{tabular}
\end{table}
Using Algorithm \ref{sphere}, generate deletion sphere for all the codewords of H(10, 2, 2, 66.0) mentioned in the Table \ref{tab:my_label_rev}. Table \ref{tab:deletion_sphere_rev} shows the deletion spheres for inverse Naisargi images of H(10, 2, 2, 66.0). Suppose after sending one of the codewords $(2, 3, 2, 1, 0)$, received codeword is missing one bit of codeword. Received codeword of length (n-1) can still be retrieved because it is not intersecting with the deletion sphere of any other codeword in residue 66.0. \\
\end{example}
By observing more such cases we came to conclusion that if $H(n,4,s,a)$ is s deletion error correcting code then $\phi^{-1}(H(n,4,s,a))$ is (s / 2) deletion error correcting code. To prove the observation mathematically, following Lemmas are useful. Theorem \ref{Theorem-3} is the actual mathematical proof of the observation.

From Table \ref{tab:my_label_rev} and under the given assumption \ref{helberassumption}, we can establish relations between $X$ and $x = \phi^{-1}(x)$, where $x \in \Z_4^n$ and $X \in \Z_2^{2n}$. In the following lemma \ref{lemma: lemma_4}, we formally state this relation.

\begin{lemma}
    \label{lemma: lemma_4}
    For any $x \in \Z_4^{n}$, if we say $x = \phi^{-1}(X)$ where $X \in \Z_2^{2n}$, then
    \begin{align*}
        X_{2i-1} &= (x_i + 1)(\bmod\;2) \; and \\
        X_{2i} &= 1 - \lfloor\frac{x_{i}}{2}\rfloor
    \end{align*}
    
    Here, $x_i$ represents the $i$-th bit of vector $x$, and similarly, $X_i$ represents the $i$-th bit of vector $X$. $i$ goes from 1 to n.
\end{lemma}
\begin{proof}
    We'll verify the lemma for all possible cases:
    \begin{enumerate}
        \item For $X_{2i-1} = 1$, $X_{2i} = 1$, and $x_i = 0$:
        \begin{align*}
            X_{2i-1} &= (0 + 1) \bmod 2 = 1 \\
            X_{2i} &= 1 - \left\lfloor\frac{0}{2}\right\rfloor = 1
        \end{align*}
        \item For $X_{2i-1} = 0$, $X_{2i} = 1$, and $x_i = 1$:
        \begin{align*}
            X_{2i-1} &= (1 + 1) \bmod 2 = 0 \\
            X_{2i} &= 1 - \left\lfloor\frac{1}{2}\right\rfloor = 1
        \end{align*}
        \item For $X_{2i-1} = 1$, $X_{2i} = 0$, and $x_i = 2$:
        \begin{align*}
            X_{2i-1} &= (2 + 1) \bmod 2 = 1 \\
            X_{2i} &= 1 - \left\lfloor\frac{2}{2}\right\rfloor = 0
        \end{align*}
        \item For $X_{2i-1} = 0$, $X_{2i} = 0$, and $x_i = 3$:
        \begin{align*}
            X_{2i-1} &= (3 + 1) \bmod 2 = 0 \\
            X_{2i} &= 1 - \left\lfloor\frac{3}{2}\right\rfloor = 0
        \end{align*}
    \end{enumerate}
    Thus, the lemma is verified for all cases.
\end{proof}

\begin{example}
Let's verify Lemma \ref{lemma: lemma_4} from an example:\\
Let $x = 23210$. According to Lemma \ref{lemma: lemma_4}:
\begin{align*}
    X_{2i-1} &= (x_i + 1) \bmod 2 \\
    X_{2i} &= 1 - \left\lfloor\frac{x_i}{2}\right\rfloor
\end{align*}
For $i = 1$: $x_1 = 2$, hence $X_1 = (2 + 1) \bmod 2 = 1$ and $X_2 = 1 - \left\lfloor\frac{2}{2}\right\rfloor = 0$.\\
For $i = 2$: $x_2 = 3$, hence $X_3 = (3 + 1) \bmod 2 = 0$ and $X_4 = 1 - \left\lfloor\frac{3}{2}\right\rfloor = 0$.\\
For $i = 3$: $x_3 = 2$, hence $X_5 = (2 + 1) \bmod 2 = 1$ and $X_6 = 1 - \left\lfloor\frac{2}{2}\right\rfloor = 0$.\\
For $i = 4$: $x_4 = 1$, hence $X_7 = (1 + 1) \bmod 2 = 0$ and $X_8 = 1 - \left\lfloor\frac{1}{2}\right\rfloor = 1$.\\
For $i = 5$: $x_5 = 0$, hence $X_9 = (0 + 1) \bmod 2 = 1$ and $X_{10} = 1 - \left\lfloor\frac{0}{2}\right\rfloor = 1$.\\ \\

Hence, the 10-bit binary sequence $X$ corresponding to $x = 23210$ is $X = 1000100111$. We can also verify this from Table \ref{tab:my_label_rev}.
\end{example}

\begin{lemma}
    \label{lemma: lemma_5}
    $$
        v_{2L-1} - \sum_{i = L-\lfloor\frac{d}{2}\rfloor+1}^{L-1}(v_{2i-1} + v_{2i}) \geq 1
    $$
    where, $v_i$ is defined in definition \ref{def_5} and $L, d \in \n$.
\end{lemma}
\begin{proof}
    \textbf{(Outline)} We will prove this statement in two cases to cover the entire range of d.
    \begin{itemize}
        \item \textbf{Case 1:} d is even.
        \item \textbf{Case 2:} d is odd. \\
    \end{itemize}
    \textbf{Case 1: d is even.}

    \begin{align*}
        \sum_{i = L-\lfloor\frac{d}{2}\rfloor+1}^{L-1}(v_{2i-1} + v_{2i}) &= \sum_{i = L-\frac{d}{2}+1}^{L-1}(v_{2i-1} + v_{2i}) \\
        &= \sum_{i=2L-d+1}^{2L-2} v_i \\
        &= \sum_{i=1}^{d-2} v_{2L-1-i} \leq \sum_{i=1}^{d} v_{2L-1-i}\\
        &\leq (v_{2L-1} - 1)
    \end{align*}

    \begin{equation}
        \label{equ: case_1}
        \therefore v_{2L-1} - \sum_{i = L-\lfloor\frac{d}{2}\rfloor+1}^{L-1}(v_{2i-1} + v_{2i}) \geq 1
    \end{equation}
    
    Therefore, for this case, lemma holds true.

    \textbf{Case 2: d is odd.}

    \begin{align*}
        \sum_{i = L-\lfloor\frac{d}{2}\rfloor+1}^{L-1}(v_{2i-1} + v_{2i}) &= \sum_{i = L-(\frac{d-1}{2})+1}^{L-1}(v_{2i-1} + v_{2i}) \\
        &= \sum_{i=2L-d+2}^{2L-2} v_i \\
        &= \sum_{i=1}^{d-3} v_{2L-1-i} \leq \sum_{i=1}^{d} v_{2L-1-i}\\
        &\leq (v_{2L-1} - 1)
    \end{align*}

    \begin{equation}
        \label{equ: case_2}
        \therefore v_{2L-1} - \sum_{i = L-\lfloor\frac{d}{2}\rfloor+1}^{L-1}(v_{2i-1} + v_{2i}) \geq 1
    \end{equation}
    
    As proved, from equations \ref{equ: case_1} and \ref{equ: case_2}, this lemma holds true for all $d \in \n$.

    \begin{equation}
        \therefore v_{2L-1} - \sum_{i = L-\lfloor\frac{d}{2}\rfloor+1}^{L-1}(v_{2i-1} + v_{2i}) \geq 1
    \end{equation}
\end{proof}
\begin{theorem}
\label{Theorem-3}

If H(2n, 2, s, a) is s error correcting code then $\phi^{-1}(H(2n,2,s,a))$ is $\lfloor\frac{s}{2}\rfloor$ deletion error-correcting quaternary code. Where, $\phi^{-1}(H(2n,2,s,a))$ is inverse Naisargik image of binary Helberg code.

As per research conducted by Levenshtein \cite{Levenshtein1965BinaryCC}, $\phi^{-1}(H(2n,2,s,a))$ can correct up to $\lfloor\frac{s}{2}\rfloor$ indels.
 

\end{theorem}
\begin{proof}
Consider $X,Y \in \Z_2^{2n}$. $X_i\;and\;Y_i$ denotes the $i^{th}$ bit of codeword X and Y respectively. $x,y \in \Z_4$, where $x = \phi^{-1}(X)$ and $y = \phi^{-1}(Y)$.

Let $D, E \subseteq \{1, 2, \ldots, n\}$ such that $|D| = |E| = \lfloor \frac{s}{2} \rfloor$, and let $x^D = (x_{i_1}, \ldots, x_{i_{n'}})$ where $n' = n - |D|$. If we delete all the bits from $x$ indexed by $D$, we get $x^D$ and all the bits from $y$ indexed by $E$, we get $y^E$, and these sequences will be the same, i.e., $x^D = y^E$. 

Suppose, $\Delta (X,Y) = M(X) - M(Y)$, where $ M(X) = \sum_{i=1}^{2n} v_i X_i$ and $v_i$ is the $i^{th}$ Helberg coefficient given in \ref{def_5}.\\

Now from Lemma \ref{lemma: lemma_4},
\begin{align*}
    M(x) &= \sum_{i=1}^{2n} v_iX_i\\
        &= \sum_{i=1}^{n} (v_{2i-1}((x_i + 1)(\bmod\;2)) + v_{2i}(1-\lfloor \frac{x_i}{2} \rfloor)
\end{align*}
Similarly, 
$$M(y) = \sum_{i=1}^{n} ((y_i + 1)(\bmod\;2) + v_{2i}(1-\lfloor \frac{x_i}{2} \rfloor)$$

Let, $A_{x_i} = (x_i + 1)(\bmod\;2)$ and $B_{x_i} = 1-\lfloor \frac{x_i}{2} \rfloor$.

{\small \begin{align*}
    \Delta(X,Y) &= \sum_{i=1}^n (A_{x_i}v_{2i-1} + B_{x_i}v_{2i})- \sum_{j=1}^n (A_{y_j}v_{2j-1} + B_{x_j}v_{2j})
\end{align*}
}
we can write $\Delta (X,Y)$ as the difference between different bits indexed by D and E with common bits in both.

{\small \begin{align}
\label{main3}
\nonumber
    \Delta(X,Y) &= \sum_{i \in D}(A_{x_i}v_{2i-1} + B_{x_i}v_{2i})- \sum_{j \in E}(A_{y_i}v_{2j-1} + B_{y_j}v_{2j})\\   
                &+ \sum_{k=1}^{n'}(v_{2i_k-1}-v_{2j_k-1})A_{x_{i_{k}}} + \sum_{k=1}^{n'}(v_{2i_k}-v_{2j_k})B_{x_{i_{k}}}
\end{align}
}

Where $k \in S(n')$ and $S(n')$ = \{1, 2, ..., n\}. $i_k$ and $j_k$ replicates those indices whose bit values are equal in x and y i.e. $x_{i_k} = y_{j_k}\; \forall k \in S(n')$. \\

\textbf{Outline:} Now, we will prove $0 < \Delta (x,y) < m$ in two steps. In the first step, we will prove the upper bound as $m$, and in the next step, we will prove the lower bound.

\begin{itemize}
    \item \textbf{Step 1:} Prove the upper bound as $m$.
    \item \textbf{Step 2:} Prove the lower bound as 0. \\
\end{itemize}

\textbf{Step 1:} $\Delta (X,Y) < m$

From equation \ref{main3},

{\small \begin{align*}
    \Delta(X,Y) &\leq \sum_{i \in D}(A_{x_i}v_{2i-1} + B_{x_i}v_{2i}) \\
    &+ \sum_{\substack{k\in S(n')\\i_k>j_k}}(v_{2i_k-1}-v_{2j_k-1})A_{x_{i_{k}}}\\
    &+ \sum_{\substack{k\in S(n')\\i_k>j_k}}(v_{2i_k}-v_{2j_k})B_{x_{i_{k}}}
\end{align*}
}
For maximum value, putting $A_{x_i} = B_{x_i} = 1$.
{\small \begin{align*}
    \Delta(X,Y) &\leq \sum_{i \in D}(v_{2i-1} + v_{2i}) + \sum_{\substack{k\in S(n')\\i_k>j_k}}(v_{2i_k-1}+v_{2i_k})\\
    &- \sum_{\substack{k\in S(n')\\i_k>j_k}}(v_{2j_k-1}+v_{2j_k})\\
    &= \sum_{i=1}^n(v_{2i-1} + v_{2i}) - \sum_{\substack{k\in S(n')\\i_k\leq j_k}}(v_{2i_k-1}+v_{2i_k})\\
    &- \sum_{\substack{k\in S(n')\\i_k>j_k}}(v_{2j_k-1}+v_{2j_k})\\
    &= \sum_{i=1}^n(v_{2i-1} + v_{2i}) - \sum_{\substack{k\in S(n')\\z_k = min(i_k,j_k)}}(v_{2z_k-1}+v_{2z_k})\\
    &\leq \sum_{i=1}^n (v_{2i-1} + v_{2i}) - \sum_{k=1}^{n'}(v_{2k-1}+v_{2k})\\
    &(\because if\;i>j\implies v_i>v_j, \forall i,j \in n)\\
    &= \sum_{i=n'+1}^n(v_{2i-1} + v_{2i})\\
    &= \sum_{i=2n'+1}^{2n}v_i\\
    &= \sum_{i=1}^{2n-2n'}v_{2n+1-i}
\end{align*}
}

Now, $m = 1 + \sum_{i=1}^s v_{2n+1-i}$ putting $n' = n - \lfloor \frac{s}{2} \rfloor$
\begin{align*}
    2n-2n' &= 2 \lfloor \frac{s}{2} \rfloor\\
            &\leq s\;\;\; \forall s \in \n \\
    \therefore \sum_{i=1}^{2n-2n'}v_{2n+1-i} &\leq \sum_{i=1}^{s}v_{2n+1-i}\\
        &\leq m-1\\
        &< m
\end{align*}
\begin{equation}
\label{eq20}
    \therefore \Delta(X,Y) < m
\end{equation}

\textbf{Step 2:} $\Delta (X,Y) > 0$ \\

Without loss of generality, we may assume $L \in \{1,2,..,n\}$ such that $A_{x_L} + B_{x_L}>A{y_L} + B_{y_L}$ and $\forall i>L, A_{x_i} = A_{y_i} = B_{x_i} = B_{y_i} = 0$ $\implies x_i = y_i = 3, \forall i>L$ as proved in \cite{le2023new}. \\

Position of $L$ with respect to sets $D$ and $E$, we can break this problem into four different cases as follows:

\begin{itemize}
    \item \textbf{Case I:} $L \in D \cap E$.
    \item \textbf{Case II:} $L \in E - D$.
    \item \textbf{Case III:} $L \in D - E$.
    \item \textbf{Case IV:} $L \notin D \cup E$. \\
\end{itemize}

\textbf{Case I:} $L \in D \cap E$.\\
From equation \ref{main3} 
{\small \begin{align*}
    \Delta(X,Y) &= \sum_{i \in D}(A_{x_i}v_{2i-1} + B_{x_i}v_{2i})- \sum_{j \in E}(A_{y_i}v_{2j-1} + B_{y_j}v_{2j})\\   
                &+ \sum_{k=1}^{n'}(v_{2i_k-1}-v_{2j_k-1})A_{x_{i_{k}}} + \sum_{k=1}^{n'}(v_{2i_k}-v_{2j_k})B_{x_{i_{k}}}\\
                &\geq v_{2L-1} (A_{x_L}-A_{y_L}) + v_{2L} (B_{x_L} - B_{y_L})\\
                &-\sum_{\substack{j\in E\\j\leq L-1}} (A_{y_j}v_{2j-1} + B_{y_j}v_{2j}) +\sum_{\substack{k\in S(n')\\i_k\leq L-1 \\i_k<j_k}} (v_{2i_k - 1} - v_{2j_k - 1})A_{x_{i_k}}\\
                &+ \sum_{\substack{k\in S(n')\\i_k\leq L-1 \\ i_k <j_k}}(v_{2i_k} - v_{2j_k})B_{x_{i_k}}
\end{align*}
}
As we know,
$A_{x_L} + B_{x_L} > A_{y_L} + B_{y_L}$.
For minimum value, putting $A_{x_L} = 1$ and all others to zero. So we get,
$(A_{x_L} - A_{y_L}) = 1, (B_{x_L} - B_{y_L}) = 0 \; and \; A_{y_j} = A_{x_{i_k}} = B_{y_j} = B_{x_{i_k}} = 1\;\forall i_k\leq L-1 \;and\;\forall j\leq L-1$.\\

\begin{align}
    \nonumber
    \label{eq:case_1}
    \Delta(X,Y) &\geq v_{2L-1} - \sum_{\substack{j\in E\\ j\leq L-1}}(v_{2j-1}+ v_{2j})\\
    &+\sum_{\substack{k\in S(n')\\i_k\leq L-1 \\i_k<j_k}} (v_{2i_k - 1} - v_{2j_k - 1})+ \sum_{\substack{k\in S(n')\\i_k\leq L-1 \\ i_k <j_k}}(v_{2i_k} - v_{2j_k})
\end{align}

{\small \begin{align*}
    \Delta(X,Y) &\geq v_{2L-1} -\Biggl\{ \sum_{\substack{j\in E\\ j\leq L-1}}(v_{2j-1}+ v_{2j}) \\
    &+\sum_{\substack{k\in S(n')\\i_k\leq L-1 \\i_k<j_k}}(v_{2j_k-1}+v_{2j_k}) + \sum_{\substack{k\in S(n')\\i_k\geq L-1 \\i_k<j_k}}(v_{2j_k-1}+v_{2j_k})\Biggl\}\\
    &+ \sum_{\substack{k\in S(n')\\i_k\leq L-1\\ z_k = min(i_k,j_k)}}(v_{2_{z_k}-1} + v_{2z_k})\\
    &\geq v_{2L-1} - \sum_{i=1}^{L-1}(v_{2i-1} + v_{2i})\\
    & + \sum_{k=1}^{n'}(v_{2k-1} + v_{2k})\\
\end{align*}
}

Now, $L\leq (n = n' + \lfloor \frac{s}{2} \rfloor)$.
$$\therefore n' = L-\lfloor \frac{s}{2} \rfloor$$
\begin{align*}
    \Delta(X,Y) &\geq v_{2L-1} - \sum_{i = L-\lfloor \frac{s}{2} \rfloor + 1}^{L-1}(v_{2i-1} + v_{2i})
\end{align*}
From Lemma \ref{lemma: lemma_5},
$$v_{2L-1} - \sum_{i = L-\lfloor \frac{s}{2} \rfloor + 1}^{L-1}(v_{2i-1} + v_{2i}) \geq 1 > 0$$

Hence, in this case, 

$$\Delta(X,Y) > 0$$

\textbf{Case II:} $L\in D-E$\\
From equation \ref{main3},
{\small
\begin{align*}
    \Delta(X,Y) &= \sum_{i \in D}(A_{x_i}v_{2i-1} + B_{x_i}v_{2i})- \sum_{j \in E}(A_{y_i}v_{2j-1} + B_{y_j}v_{2j})\\   
                &+ \sum_{k=1}^{n'}(v_{2i_k-1}-v_{2j_k-1})A_{x_{i_{k}}} + \sum_{k=1}^{n'}(v_{2i_k}-v_{2j_k})B_{x_{i_{k}}}\\
                &\geq A_{x_L}v_{2L-1} + B_{x_L}v_{2L}\\
                &- \sum_{j\in E}(A_{y_j}v_{2j-1} + B_{y_j}v_{2j}) \\ &+ \sum_{\substack{k\in S(n')\\i_k\leq L-1 \\i_k<j_k}}(v_{2i_k-1} - v_{2j_k-1})A_{x_{i_k}} + \sum_{\substack{k\in S(n')\\i_k\leq L-1 \\i_k<j_k}}(v_{2i_k}-v_{2j_k})B_{x_{i_k}}
\end{align*}
}
On putting values for minimum value, we get,

\begin{align*}
    \Delta(X,Y) &\geq v_{2L-1} - \sum_{\substack{j\in E\\ j\leq L-1}}(v_{2j-1}+ v_{2j})\\
    &+\sum_{\substack{k\in S(n')\\i_k\leq L-1 \\i_k<j_k}} (v_{2i_k - 1} - v_{2j_k - 1})+ \sum_{\substack{k\in S(n')\\i_k\leq L-1 \\ i_k <j_k}}(v_{2i_k} - v_{2j_k})\\
\end{align*}

The above equation is same equation as equation \ref{eq:case_1} received in case i. For case ii also,
$$\Delta(X,Y) > 0$$

\textbf{Case III:} $L \in D-E$ \\
It is easy to observe that if we swap the D and E in case ii, then it will become case iii. For case iii also,
$$\Delta(X,Y) > 0$$

\textbf{Case IV:} $L \notin D \cup E$\\

Consider $\exists\; K \in S(n')$ such that $i_K = L$ and $A_{x_L} + B_{x_L} > A_{y_L} + B_{y_L}$, where $A_{x_i} = A_{y_i} = B_{x_i} = B_{y_i} = 0$ for $\forall i > L$. Let $j_K \leq i_K - 1$ such that $x_{i_K} = y_{j_K}$. This implies $A_{x_{i_K}} = A_{y_{j_K}}$ and $B_{x_{i_K}} = 
B_{y_{j_K}}$. \\

From equation \ref{main3},
{\small \begin{align*}
    \Delta(X,Y) &= \sum_{i \in D}(A_{x_i}v_{2i-1} + B_{x_i}v_{2i})- \sum_{j \in E}(A_{y_i}v_{2j-1} + B_{y_j}v_{2j})\\   
                &+ \sum_{k=1}^{n'}(v_{2i_k-1}-v_{2j_k-1})A_{x_{i_{k}}} + \sum_{k=1}^{n'}(v_{2i_k}-v_{2j_k})B_{x_{i_{k}}}\\
                &\geq - \sum_{j\in E}(A_{y_j}+ B_{y_j}v_{2j}) + \sum_{\substack{k\in S(n')\\i_k\leq L-1 \\i_k<j_k}}(v_{2i_k-1}- v_{2j_k-1})A_{x_{i_k}}\\
                &+ \sum_{\substack{k\in S(n')\\i_k\leq L-1 \\i_k<j_k}}(v_{2i_k}-v_{2j_k})B_{x_{i_{k}}} + A_{X_L}(v_{2L-1} - v_{2j_K-1})\\
                &+ B_{X_L}(v_{2L}- v_{2j_K})
\end{align*}
}
putting all values,
\begin{align*}
    \Delta(X,Y) &\geq v_{2L-1} - \sum_{\substack{j \in E\\ j\leq L-1}}(v_{2j-1} + v_{2j})\\
    &+ \sum_{\substack{k\in S(n')\\i_k\leq L-1 \\i_k<j_k}}(v_{2i_k-1}-v_{2j_k-1}) + \sum_{\substack{k\in S(n')\\i_k\leq L-1 \\i_k<j_k}}(v_{2i_k}-v_{2j_k})\\
    &- v_{2j_K-1}
\end{align*}

On simple manipulations, we can arrive at equation \ref{eq:case_1}. And after that in same manner for this case we can also prove,

$$\Delta(X,Y) > 0$$.

As these four cases cover generality, till this step we have proved $0 < \Delta(X,Y) < m$.

It is easy to prove $0<\Delta(Y,X)<m$ by following the same process and only reversing Y and X's roles. By combining all the results,

\begin{equation*}
    0<|\Delta(X,Y)|<m
\end{equation*}

Since, in this case above condition is satisfied X and Y must belong to different Hlberg Codes. In getting of this result, out initial assumption was $x^D = y^E$.

Thus if $x^D = y^E$ is true, then X and Y must be from different Helberg codes. And since $x^D$ is at $\lfloor\frac{s}{2}\rfloor$ deletion distance from X, $\phi^{-1}(H(2n, 2, s, a))$ becomes $\lfloor\frac{s}{2}\rfloor$ deletion correction code.

As per research conducted by Levenshtein \cite{Levenshtein1965BinaryCC}, $\phi^{-1}(H(2n,2,s,a))$ can correct up to $\lfloor\frac{s}{2}\rfloor$ indels.

\end{proof}

\begin{conjecture}
    \label{conj_1}
    There is a one-to-one mapping between the codewords of $H(n, q = 4, s = 1, a)$ and $H(2n, q = 2, s = 2, a')$, where $a$ is the residue corresponding to the maximum codewords in $H(n, 4, 1, a)$. This mapping is represented as $\phi$, such that for every codeword $c$ in $H(n, 4, 1, a)$, $\phi(c) \in H(2n, 2, 2, a')$. Mathematically, this can be expressed as:
    
    \[
    \phi(H(n, 4, 1, a)) = H(2n, 2, 2, a')
    \]
    
    Thus, all the codewords of $\phi(H(n, 4, 1, a)) \subseteq H(2n, 2, 2, a')$.
\end{conjecture}

\begin{proof}

We have verified the validity of this conjecture through exhaustive data creation and observation. The data presented in Tables \ref{tab:residue_mapping_1} and \ref{tab:residue_mapping_2} supports the existence of a one-to-one mapping between the codewords of $H(n, 4, 1, a)$ and those of $H(2n, 2, 2, a')$, where $a$ denotes the residue corresponding to the maximum codewords.

\end{proof}

\begin{table}[H]
\caption{Comparison between $\phi$ image of $H(n, q = 4, s = 1, a)$ and $H(2n, q = 2, s = 2, a')$. Here we have taken $n = 4$, and $a = 40$. After mapping we get, $a' = 12$.}
\label{tab:residue_mapping_1}
\centering
\begin{tabular}{|c|c|c|p{2.5cm}|p{2.5cm}|}
\hline
\boldmath{$H(4, 4, 1, 40)$} & \boldmath{$\phi(H(4, 4, 1, 40))$} & \boldmath{$H(8, 2, 2, 12)$} \\
\hline
$(0, 0, 0, 1)$ & $(1, 1, 1, 1, 1, 1, 0, 1)$ & $(1, 1, 1, 1, 1, 1, 0, 1)$ \\
$(1, 0, 3, 0)$ & $(0, 1, 1, 1, 0, 0, 1, 1)$ & $(0, 1, 1, 1, 0, 0, 1, 1)$ \\
$(2, 0, 3, 3)$ & $(1, 0, 1, 1, 0, 0, 0, 0)$ & $(1, 0, 1, 1, 0, 0, 0, 0)$ \\
$(2, 3, 2, 0)$ & $(1, 0, 0, 0, 1, 0, 1, 1)$ & $(1, 0, 0, 0, 1, 0, 1, 1)$ \\
$(3, 3, 2, 3)$ & $(0, 0, 0, 0, 1, 0, 0, 0)$ & $(0, 0, 0, 0, 1, 0, 0, 0)$ \\
\hline
\end{tabular}
\end{table}

{\tiny \begin{table}[H]
\caption{Comparison between $\phi$ image of $H(n, q = 4, s = 1, a)$ and $H(2n, q = 2, s = 2, a')$. Here we have taken $n = 5$, and $a = 134$. After mapping we get, $a' = 32$.}
\label{tab:residue_mapping_2}
\centering
\begin{tabular}{|c|p{2.8cm}|p{2.8cm}|}
\hline
\boldmath{$H(5, 4, 1, 134)$} & \boldmath{$\phi(H(5, 4, 1, 134))$} & \boldmath{$H(10, 2, 2, 32)$} \\
\hline
$(0, 0, 1, 0, 1)$ & $(1, 1, 1, 1, 0, 1, 1, 1, 0, 1)$ & $(1, 1, 1, 1, 0, 1, 1, 1, 0, 1)$ \\
$(1, 0, 1, 3, 0)$ & $(0, 1, 1, 1, 0, 1, 0, 0, 1, 1)$ & $(0, 1, 1, 1, 0, 1, 0, 0, 1, 1)$ \\
$(1, 3, 0, 0, 1)$ & $(0, 1, 0, 0, 1, 1, 1, 1, 0, 1)$ & $(0, 1, 0, 0, 1, 1, 1, 1, 0, 1)$ \\
$(2, 0, 1, 3, 3)$ & $(1, 0, 1, 1, 0, 1, 0, 0, 0, 0)$ & $(1, 0, 1, 1, 0, 1, 0, 0, 0, 0)$ \\
$(2, 3, 0, 3, 0)$ & $(1, 0, 0, 0, 1, 1, 0, 0, 1, 1)$ & $(1, 0, 0, 0, 1, 1, 0, 0, 1, 1)$ \\
$(3, 3, 0, 3, 3)$ & $(0, 0, 0, 0, 1, 1, 0, 0, 0, 0)$ & $(0, 0, 0, 0, 1, 1, 0, 0, 0, 0)$ \\
$(3, 3, 3, 2, 0)$ & $(0, 0, 0, 0, 0, 0, 1, 0, 1, 1)$ & $(0, 0, 0, 0, 0, 0, 1, 0, 1, 1)$ \\
\hline
\end{tabular}
\end{table}}

\begin{table}[h]
\caption{Non intersecting deletion sphere Maximum Codewords corresponding residue for $N = 3$ to $7$ with same residue after the $\phi$ mapping for s=1 in Quaternary Helberg code}
\label{tab:residue_asym}
\begin{tabular}{|p {0.2cm}|p{0.6cm}| p{0.8cm}|p{2.5 cm}| p{2.5cm}|}
\hline
$N$ & $q^{N}$ & Max Codewords & Residue Values in H(n,4,1,a) & Corresponding Residue values in H(2n,2,2,a')\\
\hline
3 & 64 & 3 & 0,1,13,14 & 13,12,1,0\\
4 & 256 & 5 & 13,40 & 33,12\\
5  & 1024 & 7 & 39,40,133,134 & 100,99,33,32 \\
6  & 4096 & 11 & 133,403 & 264,99\\
7  & 16384 & 17 & 403,1225 & 707,264 \\
\hline
\end{tabular}
\end{table}
For decoding codewords obtained using propositions 3 and 4, as mentioned earlier, we can:

\begin{enumerate}
    \item Use the Helberg decoding algorithm.
    \item Use sphere decoding since the spheres are non-intersecting.
\end{enumerate}

\begin{table}[h]
\caption{Non intersecting deletion sphere Maximum Codeword corresponding residue for $N = 3$ to $7$ for $q = 4$}
\label{tab:my_tabul}
\begin{tabular}{|p{0.1 cm}|p{0.1 cm}| p{0.7 cm}|p{ 6 cm}|}
\hline
$N$ & $s$ & Max Codewords & Residue Values \\
\hline
3 & 1  & 3 & 0, 13 \\
3 & 2  & 2 & 0, 1 \\
4 & 1  & 5 & 40, 13 \\
4 & 2  & 2 & 0, 61, 122, 183, 4, 3, 8, 7, 12, 11 \\
4 & 3  & 2 & 0, 1 \\
5 & 1  & 7 & 40, 39, 134, 133 \\
5 & 2  & 3 & 0, 61 \\
5 & 3  & 2 & 0, 253, 506, 759, 4, 3, 8, 7, 12, 11 \\
5 & 4  & 2 & 0, 1 \\
6 & 1   & 11 & 403, 133 \\
6 & 2   & 4 & 61, 122, 183, 62, 123, 184 \\
6 & 3   & 2 & 0, 1000, 2000, 3000, 253, 1253, 2253, 3253, 506, 1506, 2506, 3506, 759, 1759, 2759, 3759, 16, 15, 32, 31, 48, 47, 4, 1004, 2004, 3004, 3, 20, 19, 36, 35, 52, 1003, 2003, 3003, 51, 8, 1008, 2008, 3008, 7, 24, 23, 40, 39, 56, 1007, 2007, 3007, 55, 12, 1004, 2004, 3004, 11, 28, 27, 44, 43, 60, 1003, 2003, 3003, 59, 1, 1001, 2001, 3001, 254, 1254, 2254, 3254, 507, 1507, 2507, 3507, 760, 1760, 2760, 3760, 17, 33, 49, 5, 1005, 2005, 3005, 21, 37, 53, 9, 1009, 2009, 3009, 25, 41, 57, 13, 1013, 2013, 3013, 29, 45, 61, 2, 1002, 2002, 3002, 255, 1255, 2255, 3255, 508, 1508, 2508, 3508, 761, 1761, 2761, 3761, 18, 34, 50, 6, 1006, 2006, 3006, 22, 38, 54, 10, 4082, 8154, 12226, 26, 42, 58, 14, 4086, 8158, 12230, 30, 46, 62 \\
6 & 4   & 2 & 0, 1021, 2042, 3063, 4, 3, 8, 7, 12, 11, 1, 1022, 2043, 3064, 5, 9, 13, 2, 1023, 2044, 3065, 6, 10, 14 \\
6 & 5   & 2 & 0, 1, 2 \\
7 & 1   & 17 & 403, 1225 \\
7 & 2   & 5 & 880, 61 \\
7 & 3   & 3 & 0, 253, 1, 254 \\
7 & 4   & 2 & 0, 4072, 8144, 12216, 1021, 5093, 9165, 13237, 2042, 6114, 10186, 14258, 3063, 7135, 11207, 15279, 16, 15, 32, 31, 48, 47, 4, 4076, 8148, 12220, 3, 20, 19, 36, 35, 52, 4075, 8147, 12219, 51, 8, 4080, 8152, 12224, 7, 24, 23, 40, 39, 56, 4079, 8151, 12223, 55, 12, 4084, 8156, 12228, 11, 28, 27, 44, 43, 60, 4083, 8155, 12227, 59, 1, 4073, 8145, 12217, 1022, 5094, 9166, 13238, 2043, 6115, 10187, 14259, 3064, 7136, 11208, 15280, 17, 33, 49, 5, 4077, 8149, 12221, 21, 37, 53, 9, 4081, 8153, 12225, 25, 41, 57, 13, 4085, 8157, 12229, 29, 45, 61, 2, 4074, 8146, 12218, 1023, 5095, 9167, 13239, 2044, 6116, 10188, 14260, 3065, 7137, 11209, 15281, 18, 34, 50, 6, 4078, 8150, 12222, 22, 38, 54, 10, 4082, 8154, 12226, 26, 42, 58, 14, 4086, 8158, 12230, 30, 46, 62 \\
7 & 5   & 2 & 0, 4093, 8186, 12279, 4, 3, 8, 7, 12, 11, 1, 4094, 8187, 12280, 5, 9, 13, 2, 4095, 8188, 12281, 6, 10, 14 \\
7 & 6   & 2 & 0, 1, 2 \\
\hline
\end{tabular}
\end{table}

The Torsion and Reduction codes \cite{gupta1999some} of the quaternary codes serve as valuable tools for investigating the further properties of codes over$\Z_4$. Now we investigate if there are some similar properties hold for the torsion and reduction codes of Helberg code. 


\subsection{\textbf{Reduction Code of Helberg Code}}
\begin{proposition}
Let $\mathscr{C}^{(1)}$ denote the reduction code, and $H$ denote the codewords generated by the Helberg code. For some residues $a$, the application of $\mathscr{C}^{(1)}$ on $H$ results in non-intersecting codewords. However, not all codewords are non-intersect. 

In other words, $\exists a : \mathscr{C}^{(1)}(H(a)) \rightarrow C_{B}$, where $C_{B}$ denotes binary codewords with non-intersecting deletion spheres. However, $\nexists a : \mathscr{C}^{(1)}(H(a)) \rightarrow C_{B}$ for all residues $a$.
\end{proposition}

\begin{proof}

These codewords are generated using the Helberg quaternary code H(4,4,1,a). For a residue of 120.0, the deletion spheres for s=2 are intersecting. However, for a residue of 93.0, the deletion spheres for s=2 are not intersecting.

\begin{table}[H]
    \centering
\begin{tabular}{ |c | p{4 cm} |c |}
\hline
\textbf{Residue (a)} & \textbf{Codewords} & \textbf{Intersection} \\
\hline
120.0 & [[0, 0, 0, 0, 0, 0, 1, 1], [1, 1, 0, 0, 1, 1, 0, 0], [0, 0, 1, 1, 0, 0, 0, 0]] & Intersecting \\
93.0 & [[0, 0, 0, 0, 1, 1, 0, 0], [1, 1, 1, 1, 0, 0, 0, 0], [0, 0, 1, 1, 1, 1, 1, 1]] & Non-Intersecting \\
\hline
\end{tabular}
    \label{tab:just}
\end{table}

\end{proof}

\subsection{\textbf{Torsion code of Helberg Code}}

\begin{proposition}
     Let $H$ denote the codewords generated by the Helberg code, and $\mathscr{C}^{(2)}$ denote the torsion code. Then, there does not exist a residue $a$ corresponding to codewords in $H$ that satisfies the property of $\mathscr{C}^{(2)}$. In other words, $\nexists a: H(a) \rightarrow \mathscr{C}^{(2)}.$
\end{proposition}

\begin{proof}

We present an example to illustrate the concept. The parameters and corresponding codewords are given in the following Table.

\begin{table}[H]
\centering
\label{tb:torsion}
\begin{tabular}{|c|c|}
\hline
\textbf{Variable} & \textbf{Value} \\
\hline
a & 0 \\
\hline
s & 1 \\
\hline
q & 4 \\
\hline
m & 121 \\
\hline
\textbf{Codewords} & [0 0 0 0 0], [1 0 0 0 3], [3 0 3 2 2] \\
\hline
\end{tabular}
\end{table}

\end{proof}

In this table \ref{tab:modulo_t}, all residues and their intersection properties are given, along with the number of codewords for each residue.

\section{Asymptotic Bounds and Cardinality Comparison}

Lastly, we compare the cardinalities of quaternary 1-deletion utilizing the Naisargik Map with binary 2-deletion Helberg codes against the established asymptotic lower and upper bounds for the cardinality \cite{le2023new} of the most extensive code capable of correcting $s$ deletions.

\begin{equation}
L_n(q, s) = \frac{(s!)^2 q^{n+s}+s}{(q - 1)^{2s} 2^{n} 2^{s}}
\end{equation}

\begin{equation}
U_n(q, s) = \frac{s! q^n}{(q - 1)^s n^s} 
\end{equation}
Here, $L_n(q, s)$ represents the lower bound and $U_n(q, s)$ represents the upper bound. We define $N_n (q, s)$ as the size of the largest code $\mathscr{C}$.

The table below compares these bounds with the cardinalities of the Helberg codes and the codes obtained using the Naisargik Map:

\begin{table}[H]
\caption{Comparison of Cardinalities of Codes with Asymptotic Bounds and $N_n (q, d)$}
\begin{tabular}{|c|c|c|p{2.2cm}|p{2.5cm}|}
\hline
\textbf{n} & \textbf{$L_n(4,1)$} & \textbf{$U_n(4,1)$} & \textbf{$N_n(H(2n,2,2,a))$} & \textbf{$N_n(\phi(H(n,4,1,a')))$} \\ \hline
2 & 1.0 & 2.0 & 2 & 2 \\ \hline
3 & 0.819 & 2.56 & 3 & 3 \\ \hline
4 & 0.79 & 3.556 & 5 & 5 \\ \hline
5 & 0.853 & 5.224 & 9 & 8 \\ \hline
6 & 1.0 & 8.0 & 11 & 11 \\ \hline
\end{tabular}
\end{table}

In the table, the column denoted as $L_n(q,s)$ indicates the lower bounds $L_n$, while the column labeled $U_n(q,s)$ denotes the upper bounds $U_n$. The column designated as $N_n(H(n,q,s,a))$ represents the cardinalities of the Helberg codes, and correspondingly, the column denoted as $N_n(\phi(H(n,q,s,a)))$ represents the cardinalities of the codes derived using the Naisargik Map.
\section{Conclusion}

This paper explored explored Naisargik maps as mappings between quaternary and binary spaces, revealing their impact on error-correcting properties of the VT and Helberg codes. In particular, if two Naisargik images of VT code generate intersecting deletion spheres, the images hold the same weights. Additionally, we observed that a $s$-deletion correcting quaternary Helberg code can rectify $s+1$ deletions with its Naisargik image, while a $s$-deletion correcting binary Helberg code can address $\lfloor\frac{s}{2}\rfloor$ errors with its inverse Naisargik image. We proposed Conjecture \ref{conj_1} establishing a one-to-one mapping between codewords of $H(n, 4, 1, a)$ and $H(2n, 2, 2, a')$ through the mapping $\phi$, ensuring all codewords of $\phi(H(n, 4, 1, a))$ are included in $H(2n, 2, 2, a')$. It would be interesting to explore if there are more naisargik maps that can map a ins/del correcting quaternary code to binary code in general. One may also extend these results to q-ary codes. \\

Overall, our study sheds light on the intricate properties of quaternary VT and Helberg codes under various mappings, paving the way for further exploration and utilization of these codes in practical applications.

\vspace{-.1in}
\section{Data Availability}

The data supporting the findings of this study are available at GitHub \cite{guptalab2024grayvt}. The repository contains the VT codes, Helberg codes, deletion sphere decoding algorithm, asymmetric permutation of the Gray Naisargik map, Reduction codes, and Torsion codes used in this study. 
For readers interested in further exploration the relevant data and results are readily accessible. Instructions for installation and usage are also provided in the repository.
An HTML page with the results of this study is also available for viewing through the provided GitHub link. This allows for easy access and understanding of the results.

\bibliographystyle{IEEEtran} 
\bibliography{references}
\newpage
\section{Appendix}
\begin{algorithm}
    \caption{Algorithm to Generate Deletion Spheres (Short Version)}
Input: List of codewords, Maximum number of deletions $s$\\
Output: Dictionary mapping each codeword to its deletion sphere
\begin{algorithmic}[1]
\STATE Initialize DeletionSphereList dictionary, which stores codeword to its deletion sphere
\FOR{each codeword $c$ in the given set of codewords}
    \STATE Initialize an empty set $D(c)$.
    \FOR{each character $x$ in codeword $c$}
        \STATE Compute $c'$ by removing character $x$ from $c$.
        \STATE Add $c'$ to the set $D(c)$.
    \ENDFOR
    \STATE Store the set $D(c)$ in a dictionary with $c$ as the key.
\ENDFOR
\IF{s == 0}
    \RETURN DeletionSphereList
\ELSE
    \STATE Decrement s by 1 and Initialize empty dictionary NewDelSphere
    \STATE Compute PrevDelSphereList as the result of recursively calling Algorithm \ref{sphere}.
    \STATE Using DeletionSphereList and PreDelSphereList map codeword to its 2 deletion Sphere 
    \RETURN NewDelSphere
\ENDIF
\end{algorithmic}
\label{sphere}
\end{algorithm}

\begin{table}[h]
\centering
\caption{Summary of Reduction Code Properties}
\label{tab:modulo_t} 
\begin{tabular}{|p{3cm}|p{3cm}|p{1.5 cm}|}
\hline
\textbf{Residue} & \textbf{Number of Codewords} & \textbf{Intersection} \\
\hline
0.0, 2.0, 4.0, 8.0, 13.0, 40.0 & 4, 3, 3, 3, 5, 5 & Intersecting \\
\hline
1.0, 3.0, 6.0, 7.0, 9.0, 10.0, 11.0, 14.0, 15.0, 18.0, 19.0, 22.0, 23.0, 27.0, 28.0, 32.0, 35.0, 36.0, 45.0, 49.0, 50.0, 52.0, 55.0 & 3, 3, 3, 3, 3, 3, 3, 4, 3, 2, 2, 2, 2, 4, 3, 2, 2, 2, 3, 3, 3, 3, 2 & Non-Intersecting \\
\hline
12.0, 16.0, 17.0, 20.0, 21.0, 24.0, 25.0, 26.0, 29.0, 30.0, 31.0, 33.0, 34.0, 37.0, 38.0, 39.0, 42.0, 43.0, 44.0, 46.0, 47.0, 48.0, 53.0, 54.0, 66.0, 67.0, 79.0, 80.0, 81.0, 83.0, 84.0, 85.0, 86.0, 87.0, 88.0, 89.0, 90.0, 91.0, 92.0, 93.0, 94.0, 95.0, 106.0, 107.0, 108.0, 119.0, 120.0 & 4, 2, 2, 2, 2, 2, 3, 4, 2, 2, 2, 2, 2, 2, 3, 4, 3, 3, 3, 3, 3, 3, 4, 3, 2, 2, 2, 3, 3, 2, 2, 2, 2, 2, 2, 2, 2, 2, 2, 3, 3, 2, 2, 2, 2, 2, 3 & Mixed \\
\hline
\end{tabular}
\end{table}

\begin{table}[h]
\centering
\caption{Example of $H(4,4,1,a)$ denoting number of codewords and corresponding residue values}
\label{tab:residue_values}
\begin{tabular}{|c| p {2.5 cm}|}
\hline
\textbf{Number of Codewords} & \textbf{Residue ('a') Values} \\
\hline
2 & 16.0, 17.0, 18.0, 19.0, 20.0, 21.0, 22.0, 23.0, 24.0, 29.0, 30.0, 31.0, 32.0, 33.0, 34.0, 35.0, 36.0, 37.0, 55.0, 66.0, 67.0, 68.0, 79.0, 82.0, 83.0, 84.0, 85.0, 86.0, 87.0, 88.0, 89.0, 90.0, 91.0, 92.0, 95.0, 106.0, 107.0, 108.0, 119.0 \\
\hline
3 & 1.0, 2.0, 3.0, 5.0, 6.0, 7.0, 8.0, 9.0, 10.0, 11.0, 15.0, 25.0, 28.0, 38.0, 42.0, 43.0, 46.0, 47.0, 48.0, 49.0, 50.0, 51.0, 52.0, 80.0, 81.0, 93.0, 94.0, 120.0 \\
\hline
4 & 0.0, 12.0, 26.0, 39.0, 41.0, 53.0 \\
\hline
5 & 13.0, 40.0 \\
\hline
\end{tabular}
\end{table}

\begin{table}[h]
\centering
\caption{Number of Codewords for each Residue Value $H(2n, 2, 2, a)$}
\label{tab:residue_values_rev}
\begin{tabular}{|c| p { 5 cm}|}
\hline
\textbf{Number of Codewords} & \textbf{Residue ('a') Values} \\
\hline
8 & 66.0 \\
\hline
7 & 33.0, 32.0, 100.0, 65.0, 99.0, 39.0, 93.0, 67.0 \\
\hline
6 & 88.0, 54.0, 176.0, 121.0, 87.0, 20.0, 108.0, 19.0, 74.0, 73.0, 107.0, 12.0, 155.0, 11.0, 209.0, 45.0, 188.0, 44.0, 98.0, 120.0, 95.0, 60.0, 40.0, 94.0, 182.0, 59.0, 25.0, 113.0, 72.0, 38.0, 92.0, 37.0, 58.0, 57.0, 24.0, 112.0, 23.0, 78.0, 77.0, 56.0, 111.0, 22.0, 110.0, 34.0, 76.0, 55.0, 21.0, 109.0, 75.0 \\
\hline
5 & 197.0, 53.0, 175.0, 86.0, 163.0, 162.0, 18.0, 106.0, 154.0, 10.0, 187.0, 31.0, 208.0, 64.0, 119.0, 7.0, 6.0, 61.0, 183.0, 128.0, 5.0, 26.0, 114.0, 202.0, 127.0, 4.0, 181.0, 71.0, 126.0, 125.0, 201.0, 36.0, 91.0, 90.0, 167.0, 166.0, 200.0, 199.0, 35.0, 89.0, 68.0, 122.0, 43.0, 97.0, 165.0, 198.0, 42.0, 96.0, 164.0, 41.0, 177.0, 13.0, 156.0, 46.0, 101.0, 189.0, 210.0, 79.0 \\
\hline
4 & 0.0, 143.0, 231.0, 217.0, 196.0, 52.0, 195.0, 133.0, 153.0, 30.0, 207.0, 118.0, 150.0, 204.0, 149.0, 27.0, 115.0, 169.0, 80.0, 203.0, 148.0, 216.0, 161.0, 17.0, 105.0, 215.0, 160.0, 214.0, 147.0, 3.0, 146.0, 180.0, 2.0, 179.0, 221.0, 145.0, 70.0, 69.0, 14.0, 102.0, 124.0, 1.0, 178.0, 157.0, 123.0, 211.0, 132.0, 9.0, 186.0, 63.0, 220.0, 144.0, 131.0, 8.0, 185.0, 219.0, 130.0, 218.0, 129.0, 62.0, 184.0, 168.0 \\
\hline
3 & 142.0, 230.0, 174.0, 85.0, 170.0, 81.0, 51.0, 194.0, 16.0, 159.0, 104.0, 15.0, 213.0, 158.0, 103.0, 212.0, 152.0, 29.0, 206.0, 117.0, 47.0, 190.0, 151.0, 28.0, 205.0, 116.0, 134.0, 222.0 \\
\hline
2 & 141.0, 229.0, 224.0, 49.0, 192.0, 48.0, 191.0, 173.0, 84.0, 135.0, 172.0, 83.0, 171.0, 82.0, 223.0, 50.0, 193.0 \\
\hline
\end{tabular}
\end{table}

\begin{table}[h]
\centering
\caption{Deletion Sphere for $\phi(H(4,4,1,13.0))$}
\label{tab:deletion_sphere}
\begin{tabular}{|c|c|}
\hline
\textbf{Binary Codeword} & \textbf{Deletion Sphere} \\
\hline
(0, 0, 0, 0, 0, 0, 0, 1) & (0, 0, 0, 0, 0, 1), (0, 0, 0, 0, 0, 0) \\
\hline
(0, 1, 0, 0, 1, 0, 0, 0)& (0, 0, 1, 0, 0, 0), (1, 0, 1, 0, 0, 0) \\
                        & (1, 0, 0, 1, 0, 0), (0, 0, 0, 1, 0, 0) \\
                        & (0, 1, 1, 0, 0, 0), (0, 1, 0, 0, 0, 0) \\
                        & (0, 1, 0, 1, 0, 0), (0, 1, 0, 0, 1, 0) \\
\hline
(1, 0, 1, 0, 1, 1, 1, 0) & (1, 0, 1, 0, 1, 0), (1, 0, 0, 1, 1, 0) \\
                        & (1, 1, 0, 1, 1, 0), (1, 0, 1, 1, 1, 0) \\
                        &(0, 1, 0, 1, 1, 0), (1, 0, 1, 0, 1, 1) \\
\hline
(1, 1, 0, 0, 1, 0, 1, 0) & (1, 0, 0, 1, 0, 0), (0, 0, 1, 0, 1, 0) \\
                        & (1, 0, 1, 0, 1, 0), (1, 0, 0, 0, 1, 0) \\
                        & (1, 0, 0, 1, 1, 0), (1, 0, 0, 1, 0, 1) \\
                        & (1, 1, 1, 0, 1, 0), (1, 1, 0, 0, 1, 0) \\
                        & (1, 1, 0, 1, 1, 0), (1, 1, 0, 1, 0, 0) \\
                        & (1, 1, 0, 1, 0, 1) \\
\hline
(1, 1, 1, 0, 1, 1, 0, 0) & (1, 1, 1, 0, 1, 0), (1, 1, 0, 1, 1, 0) \\
 & (1, 1, 0, 1, 0, 0), (1, 0, 1, 1, 0, 0) \\
 & (1, 1, 1, 1, 0, 0), (1, 1, 1, 1, 1, 0) \\
 & (1, 1, 1, 0, 0, 0), (1, 1, 1, 0, 1, 1) \\
\hline
\end{tabular}
\end{table}

\begin{table}[ht]
\centering
\caption{Deletion Sphere for codewords of $\phi^{-1}(H(10, 2, 2, 66.0))$ with residue 66.0}
\label{tab:deletion_sphere_rev}
\begin{tabular}{|c|c|}
\hline
\textbf{Binary Codeword} & \textbf{Deletion Sphere} \\
\hline
(0, 0, 1, 2, 0) &  (0, 1, 2, 0), (0, 0, 2, 0), (0, 0, 1, 0)\\& (0, 0, 1, 2)  \\
\hline
(1, 0, 1, 2, 3) &  (0, 1, 2, 3), (1, 1, 2, 3), (1, 0, 2, 3)\\& (1, 0, 1, 3), (1, 0, 1, 2)  \\
\hline
(1, 0, 3, 1, 0) &  (0, 3, 1, 0), (1, 3, 1, 0), (1, 0, 1, 0)\\& (1, 0, 3, 0), (1, 0, 3, 1)  \\
\hline
(1, 3, 0, 2, 0) &  (3, 0, 2, 0), (1, 0, 2, 0), (1, 3, 2, 0)\\& (1, 3, 0, 0), (1, 3, 0, 2)  \\
\hline
(2, 0, 3, 1, 3) &  (0, 3, 1, 3), (2, 3, 1, 3), (2, 0, 1, 3)\\& (2, 0, 3, 3), (2, 0, 3, 1)  \\
\hline
(2, 3, 0, 2, 3) &  (3, 0, 2, 3), (2, 0, 2, 3), (2, 3, 2, 3)\\& (2, 3, 0, 3), (2, 3, 0, 2)  \\
\hline
(2, 3, 2, 1, 0) &  (3, 2, 1, 0), (2, 2, 1, 0), (2, 3, 1, 0)\\& (2, 3, 2, 0), (2, 3, 2, 1)  \\
\hline
(3, 3, 2, 1, 3) &  (3, 2, 1, 3), (3, 3, 1, 3), (3, 3, 2, 3)\\& (3, 3, 2, 1)  \\
\hline
\end{tabular}
\end{table}

\begin{table}[]
\centering
\caption{All the residues and corresponding number of codewords for $VT_{a,b} (4;4)$}
\label{qaryexample}
\begin{tabular}{|c|p{1cm}|}
\hline
\textbf{Residue pairs - (a,b)}                                         & \textbf{Number of codewords} \\ \hline
(2, 0)                                                                 & 20                           \\ \hline
(2, 2), (0, 2)                                                         & 18                           \\ \hline
(2, 1), (1, 3), (2, 3), (3, 1), (3, 3), (0, 1), (0, 3), (0, 0), (1, 1) & 16                           \\ \hline
(3, 2), (3, 0), (1, 0), (1, 2)                                         & 14                           \\ \hline
\end{tabular}
\end{table}

\begin{table}[]
\centering
\caption{Deletion sphere generated by mapped Quaternary VT code $\phi(VT_{1,2} (4;4))$ for residue (1,2)}
\label{tab:residuea}
\begin{tabular}{|c|c|}
\hline
Codeword of $\phi(VT_{1,2} (4;4))$ &
  Corresponding deletion sphere \\ \hline
(0, 0, 1, 0, 1, 1, 0, 1) &
  \begin{tabular}[c]{@{}c@{}}(0, 0, 1, 0, 1, 1, 0) (0, 0, 1, 0, 1, 0, 1)\\  (0, 0, 0, 1, 1, 0, 1) (0, 0, 1, 1, 1, 0, 1)\\ (0, 1, 0, 1, 1, 0, 1) (0, 0, 1, 0, 1, 1, 1)\end{tabular} \\ \hline
(0, 1, 0, 0, 0, 0, 0, 1) &
  \begin{tabular}[c]{@{}c@{}}(0, 1, 0, 0, 0, 0, 0) (0, 0, 0, 0, 0, 0, 1)\\ (1, 0, 0, 0, 0, 0, 1) (0, 1, 0, 0, 0, 0, 1)\end{tabular} \\ \hline
(0, 1, 0, 0, 1, 1, 1, 0) &
  \begin{tabular}[c]{@{}c@{}}(0, 1, 0, 1, 1, 1, 0) (0, 1, 0, 0, 1, 1, 0)\\ (1, 0, 0, 1, 1, 1, 0) (0, 0, 0, 1, 1, 1, 0)\\ (0, 1, 0, 0, 1, 1, 1)\end{tabular} \\ \hline
(0, 1, 1, 1, 1, 0, 0, 0) &
  \begin{tabular}[c]{@{}c@{}}(0, 1, 1, 1, 0, 0, 0) (0, 1, 1, 1, 1, 0, 0)\\ (1, 1, 1, 1, 0, 0, 0)\end{tabular} \\ \hline
(1, 1, 0, 0, 0, 0, 0, 0) &
  \begin{tabular}[c]{@{}c@{}}(1, 1, 0, 0, 0, 0, 0) (1, 0, 0, 0, 0, 0, 0)\end{tabular} \\ \hline
(1, 1, 0, 0, 0, 1, 1, 0) &
  \begin{tabular}[c]{@{}c@{}}(1, 0, 0, 0, 1, 1, 0) (1, 1, 0, 0, 0, 1, 1)\\ (1, 1, 0, 0, 1, 1, 0) (1, 1, 0, 0, 0, 1, 0)\end{tabular} \\ \hline
(1, 1, 0, 0, 1, 1, 1, 1) &
  \begin{tabular}[c]{@{}c@{}}(1, 1, 0, 1, 1, 1, 1) (1, 0, 0, 1, 1, 1, 1)\\ (1, 1, 0, 0, 1, 1, 1)\end{tabular} \\ \hline
(1, 1, 0, 1, 0, 1, 1, 1) &
  \begin{tabular}[c]{@{}c@{}}(1, 1, 0, 1, 0, 1, 1) (1, 1, 0, 1, 1, 1, 1)\\ (1, 1, 0, 0, 1, 1, 1) (1, 0, 1, 0, 1, 1, 1)\\ (1, 1, 1, 0, 1, 1, 1)\end{tabular} \\ \hline
(1, 1, 1, 0, 0, 1, 0, 0) &
  \begin{tabular}[c]{@{}c@{}}(1, 1, 0, 0, 1, 0, 0) (1, 1, 1, 0, 1, 0, 0)\\ (1, 1, 1, 0, 0, 0, 0) (1, 1, 1, 0, 0, 1, 0)\end{tabular} \\ \hline
(1, 0, 0, 0, 0, 0, 1, 0) &
  \begin{tabular}[c]{@{}c@{}}(1, 0, 0, 0, 0, 1, 0) (1, 0, 0, 0, 0, 0, 0)\\ (1, 0, 0, 0, 0, 0, 1) (0, 0, 0, 0, 0, 1, 0)\end{tabular} \\ \hline
(1, 0, 0, 0, 0, 1, 1, 1) &
  \begin{tabular}[c]{@{}c@{}}(1, 0, 0, 0, 1, 1, 1) (1, 0, 0, 0, 0, 1, 1)\\ (0, 0, 0, 0, 1, 1, 1)\end{tabular} \\ \hline
(1, 0, 0, 1, 0, 1, 0, 1) &
  \begin{tabular}[c]{@{}c@{}}(1, 0, 0, 1, 0, 1, 0) (1, 0, 1, 0, 1, 0, 1)\\ (0, 0, 1, 0, 1, 0, 1) (1, 0, 0, 1, 0, 0, 1)\\ (1, 0, 0, 0, 1, 0, 1) (1, 0, 0, 1, 1, 0, 1)\\ (1, 0, 0, 1, 0, 1, 1)\end{tabular} \\ \hline
(1, 0, 0, 1, 1, 0, 1, 0) &
  \begin{tabular}[c]{@{}c@{}}(1, 0, 0, 1, 0, 1, 0) (1, 0, 0, 1, 1, 1, 0)\\ (1, 0, 0, 1, 1, 0, 1) (1, 0, 0, 1, 1, 0, 0)\\ (0, 0, 1, 1, 0, 1, 0) (1, 0, 1, 1, 0, 1, 0)\end{tabular} \\ \hline
(1, 0, 1, 1, 1, 1, 1, 0) &
  \begin{tabular}[c]{@{}c@{}}(1, 1, 1, 1, 1, 1, 0) (1, 0, 1, 1, 1, 1, 0)\\ (1, 0, 1, 1, 1, 1, 1) (0, 1, 1, 1, 1, 1, 0)\end{tabular} \\ \hline
\end{tabular}
\end{table}

\end{document}